\documentclass{article}
\usepackage{graphicx} 

\usepackage[margin=1in]{geometry}
\usepackage{pdflscape}
\usepackage{titlesec}
\usepackage{indentfirst}
\usepackage{MnSymbol,wasysym}
\usepackage{marvosym}
\usepackage[LGR,T1]{fontenc}
\usepackage{hyperref}
\hypersetup{
    colorlinks,
    linkcolor={blue},
    citecolor={green},
    urlcolor={blue!80!black}
}

\usepackage{physics}
\usepackage{amsmath,amsfonts, amsthm, amssymb}
\usepackage{pifont}
\usepackage{mathtools}
\usepackage{bbm}

\makeatletter
\def\bbordermatrix#1{\begingroup \m@th
  \@tempdima 4.75\p@
  \setbox\z@\vbox{%
    \def\cr{\crcr\noalign{\kern2\p@\global\let\cr\endline}}%
    \ialign{$##$\hfil\kern2\p@\kern\@tempdima&\thinspace\hfil$##$\hfil
      &&\quad\hfil$##$\hfil\crcr
      \omit\strut\hfil\crcr\noalign{\kern-\baselineskip}%
      #1\crcr\omit\strut\cr}}%
  \setbox\tw@\vbox{\unvcopy\z@\global\setbox\@ne\lastbox}%
  \setbox\tw@\hbox{\unhbox\@ne\unskip\global\setbox\@ne\lastbox}%
  \setbox\tw@\hbox{$\kern\wd\@ne\kern-\@tempdima\left[\kern-\wd\@ne
    \global\setbox\@ne\vbox{\box\@ne\kern2\p@}%
    \vcenter{\kern-\ht\@ne\unvbox\z@\kern-\baselineskip}\,\right]$}%
  \null\;\vbox{\kern\ht\@ne\box\tw@}\endgroup}
\makeatother

\usepackage{graphicx}
\usepackage{subcaption}
\usepackage{tikz}
\usepackage{pgfplots}
\pgfplotsset{compat=1.18} 
\usetikzlibrary{angles,
                graphs,
                graphs.standard, 
                quantikz2,
                decorations.pathmorphing,
                patterns, 
                patterns.meta,
                shadows}
\tikzset{snake it/.style={decorate, decoration=snake}}

\newtheorem{prop}{Proposition}[section]
\newtheorem{theorem}{Theorem}[section]
\newtheorem{corollary}{Corollary}[theorem]
\newtheorem{lemma}[theorem]{Lemma}
\newtheorem{definition}{Definition}
\newtheorem{result}{Result}[section]

\newcommand{\iu}{{i\mkern1mu}}
\DeclarePairedDelimiter{\ceil}{\lceil}{\rceil}

\usepackage{csvsimple}
\usepackage{listings}
\usepackage{minted}

\usepackage{paralist}

\usepackage{tabularx}
\usepackage{booktabs}
\usepackage{makecell}

\usepackage{amsthm}
\usepackage{mdframed}
    \newmdtheoremenv{thrm}{Theorem}
    \newmdtheoremenv{thrm*}{Theorem}
    \newmdtheoremenv{hlo}{Overview}
    \newmdtheoremenv{stm}{Statement}

\usepackage{algorithm}
\usepackage{algpseudocode}   

\newcounter{protocol}

\newenvironment{protocol}[1][]
{%
  \refstepcounter{protocol}
  \begin{center}
  \begin{tabular}{p{0.97\textwidth}}
  \hline
  \textbf{Protocol \theprotocol.%
  \if\relax\detokenize{#1}\relax\else\ #1\fi}
  \\ \hline
}
{%
  \\ \hline
  \end{tabular}
  \end{center}
}

\newcommand{\sbline}{\\[.5\normalbaselineskip]}

\DeclareMathAlphabet{\mathgtt}{LGR}{cmtt}{m}{n}

\DeclareMathOperator\supp{supp}
\DeclareMathOperator{\Var}{Var}

\usepackage[toc,page]{appendix}

\usepackage[backend=bibtex]{biblatex}
\usepackage[affil-it]{authblk}

\title{\large{\textbf{Efficient Fidelity Estimation with Few Local Pauli Measurements}}}
\author[1]{\normalsize Mingyu Sun}
\author[1,2]{\normalsize Gabriel Waite}
\author[1,2]{\normalsize Michael J. Bremner}
\author[1]{\normalsize Christopher Ferrie}
\affil[1]{\small Centre for Quantum Software and Information, School of Computer Science, \protect\\ Faculty of Engineering \& Information Technology, \protect\\ University of Technology Sydney, NSW 2007, Australia
}
\affil[2]{\small Centre for Quantum Computation and Communication Technology}

\date{}

\begin{document}
\maketitle

\begin{abstract}
As quantum devices scale, quantifying how close an experimental state aligns with a target becomes both vital and challenging. Fidelity is the standard metric, but existing estimators either require full tomography or apply only to restricted state/measurement families. Huang, Preskill, and Soleimanifar (Nature Physics, 2025) introduced an efficient certification protocol for Haar-random states using only a polynomial number of non-adaptive, single-copy, local Pauli measurements. Here, we adopt the same data collection routine but recast it as a fidelity estimation protocol with rigorous performance guarantees and broaden its applicability. We analyze the bias in this estimator, linking its performance to the mixing time $\tau$ of a Markov chain induced by the target state, and resolve the three open questions posed by Huang, Preskill, and Soleimanifar (Nature Physics, 2025). Our analysis extends beyond Haar-random states to state $t$-designs, states prepared by low-depth random circuits, physically relevant states and families of mixed states. We introduce a $k$-generalized local escape property that identifies when the fidelity estimation protocol is both efficient and accurate, and design a practical empirical test to verify its applicability for arbitrary states. This work enables scalable benchmarking, error characterization, and tomography assistance, supports adaptive quantum algorithms in high dimensions, and clarifies fundamental limits of learning from local measurements.
\end{abstract}

\clearpage
\tableofcontents
\clearpage

\section{Introduction}
Assessing the performance and reliability of quantum systems is crucial for the development of quantum technology. Fidelity estimation, which quantifies how close a prepared quantum state is to a target, remains a fundamental tool for benchmarking, certification, and characterization \cite{PhysRevLett.106.230501, Cerezo_2021}.
However, existing fidelity estimation methods face severe limitations: exponential sample complexity~\cite{PhysRevLett.106.230501, PhysRevLett.107.210404, article_Aolita}, deep quantum circuits or sophisticated experimental controls~\cite{10.1145/2897518.2897544, 10.1145/3055399.3055454, 7956181, PhysRevLett.113.190404, gupta2025singlequbitmeasurementssufficecertify, Chen_2021}, reliance on strong oracle models~\cite{gupta2025singlequbitmeasurementssufficecertify}, or applicability restricted to special states/measurements~\cite{PhysRevLett.106.230501, 10756060, PhysRevLett.120.190501, PhysRevX.8.02106, 2a6eac0a7fff47098b29d37402b877f0, gilyén2022improvedquantumalgorithmsfidelity, Cerezo_2020}. 

\paragraph{Prior work.}Direct Fidelity Estimation (DFE) estimates the fidelity between a known pure quantum state $\vert \psi \rangle$ and an unknown
state $\rho$ by importance-sampling Pauli observables~\cite{PhysRevLett.106.230501}. DFE is constant-cost for stabilizer states, but scales as $\mathcal{O}(2^n)$ for generic states due to the Pauli spectrum spreading over $2^n$ non-uniform terms. Moreover, when some importance weights are tiny, the estimator suffers "small-denominator" variance blow-up~\cite{wong2025efficientquantumtomographypolynomial}. 
Classical shadows excel in fidelity estimation with Clifford measurements; with Pauli measurements, the variance can be exponentially large for generic targets~\cite{Huang2020PredictingMP} and, unless $\vert \psi \rangle$ has an efficient classical description, the post-processing cost remains $\mathcal{O}(2^n)$ \cite{vairogs2024directfidelityestimationgeneric}. A combination of quantum amplitude estimation with classical shadows achieves a quadratic speedup, but requires deep circuits, ancilla qubits, and precise coherent control~\cite{vairogs2024directfidelityestimationgeneric}. Cha and Lee use randomized measurements with tailored post-processing to exploit structure in specific targets, yielding 1–2 orders-of-magnitude savings over DFE, but only for restricted families~\cite{cha2025efficientsamplingpaulimeasurementbased}.

Huang, Preskill, and Soleimanifar~\cite{10756060} introduced the shadow overlap certification protocol, using only a polynomially many non-adaptive single-copy (no collective measurements) Pauli measurements. Each round randomly measures $k$ qubits in random Pauli bases and the remaining $n-k$ qubits in $Z$, yielding the $k$-qubit classical shadow and the conditional state from $\vert \psi\rangle$ with which they compute one local overlap. The average of local overlaps is compared to a threshold for certification.

The authors posed three open questions~\cite{10756060}: \begin{inparaenum}[1)]
    \item Beyond Haar random states, which families of quantum states (e.g. states generated by random quantum circuits) admit efficient certification? 
    \item Can we identify states that cannot be certified using only polynomially many non-adaptive single-copy local measurements? 
    \item Can similar protocols efficiently certify certain families of mixed states?
\end{inparaenum} Note that certification is a binary test that determines whether the fidelity exceeds a threshold, while fidelity estimation seeks a precise value within a continuous range; thus is inherently harder. In principle, a certification protocol can be turned into a fidelity estimation by scanning the threshold using multiple certification tests, but this approach is inefficient and redundant. 

\paragraph{Contributions.} We adopt the same query access and data collection routine as the shadow overlap certification protocol, but we change the postprocessing and substantially extend its theoretical foundation and practical applicability:
\begin{itemize}
    \item \textbf{From certification to estimation}: We rigorously establish the protocol as a fidelity estimator, providing quantitative fidelity values rather than binary threshold tests. The estimator is biased, but we analyze and bound the bias and propose methods to partially mitigate the bias. We also support simultaneous estimation against $M$ target with sample overhead $\mathcal{O}(\ln M)$.
    \item \textbf{Proof extensions}: We extend the proof from Haar random state to $\epsilon$-approximate state $t$-designs, states prepared by certain low-depth random quantum circuits, ground state of (non-stoquastic) gapped local Hamiltonian, W states, and Dicke states. Our theoretical findings are corroborated by numerical simulations.
    \item \textbf{Practical applicability check}: We design an empirical check that diagnoses whether the estimator will be accurate for an arbitrary target state.
    \item \textbf{Extension to mixed states}: We extend the protocol to certify certain families of mixed states.
\end{itemize}
The last three points resolve the three open questions posed by Huang, Preskill, and Soleimanifar~\cite{10756060}; the first point transforms the shadow overlap from a certification tool into a fidelity estimation framework. We relate the accuracy and sample complexity of this fidelity estimator to the mixing time $\tau$ of a Markov chain induced by the target state’s distribution~\cite{10756060}. Our first theorem shows the sample complexity scales as $\mathcal{O}(\tau^2)$, and we introduce the $k$-Generalized Local Escape Property ($k$-GLEP) that guarantees $\tau$ is polynomially bounded. 

\subparagraph{Outline.} The corresponding Theorems and Result are presented in Section~\ref{Performance guarantees}, and discussed in Section~\ref{Main content}. Appendix~\ref{Observables and Fidelity} to~\ref{Mixed states} each includes the detailed analysis and proofs of the four theoretical contributions, respectively. Numerical results in Section~\ref{Result} support the theoretical findings and demonstrate the applications of this fidelity estimation protocol. Related work, variants are discussed in Section~\ref{Discussion}.

\section{Main Results}
\label{Performance guarantees}
Let $\pi$ denote the probability distribution of the target state $\ket{\psi}$ in the computational basis, i.e., $\pi(x) = \abs{\braket{x}{\psi}}^2$ for all $x \in \{0,1\}^n$.
Our main theorem establishes that fidelities with respect to multiple target states can be simultaneously estimated using a polynomial number of non-adaptive local Pauli measurements, provided the mixing times of the Markov chains with stationary distributions $\pi$ are polynomially bounded. The mixing time of a Markov chain quantifies how quickly it converges to its stationary distribution, and thus characterizes how efficiently the chain explores the state space.
\begin{theorem}[Sample Complexity of Fidelity Estimation]
\label{Main theorem}
    Let $\rho$ be an unknown $n$-qubit state and $\{ \vert \psi_i \rangle \}_{i=1}^M $ be $M$ pure target states. For each $\vert \psi_i \rangle$, let $\tau_i$ denote the mixing time of the Markov chain induced by $\vert \psi_i \rangle$. Suppose all $\tau_i$ are polynomially bounded in $n$, then local Pauli measurements on $T = \mathcal{O} \Big (2^{2k} \frac{\max_i \tau_i^2}{(c_b \epsilon)^2} \ln \Big ( \frac{2M}{\delta} \Big ) \Big ) $ samples of $\rho$ suffice to estimate fidelities $F_i= \langle \psi_i \vert \rho \vert \psi_i  \rangle$, $i= 1, \ldots, M$ up to additive error $\epsilon$ with a failure probability at most $\delta$, where $k \ll n$ is the number of qubits measured in random Pauli bases, $c_b$ accounts for the bias, $\frac{c_b}{\tau_i} \in (0,1)$.
\end{theorem}

The fidelity estimation procedure is presented in Section~\ref{Procedure} with detailed proofs given in Appendix~\ref{Observables and Fidelity} and~\ref{Bounding the Mixing Time of the Markov Chain}; we sketch an outline here. In Appendix~\ref{Observables and Fidelity} we show that the target state $\vert \psi \rangle$ is an eigenstate of the implicit observable in the fidelity estimation protocol with eigenvalue 1. That is, the implicit observable acts as an effective projector onto $\vert \psi \rangle$, thus the measured expectation provides a proxy for the fidelity. The deviation from the true fidelity is bounded by $1-\frac{1}{\tau}$, where $\tau$ is the mixing time of the induced Markov chain defined in Section~\ref{Procedure}. Namely, our protocol provides a biased estimator for fidelity. We take the bias into consideration via $\tau$ and the triangle inequality when analyzing the complexity and provide rigorous performance guarantees. We bound the bias to be at most $(1-\frac{c_b}{\tau}) \epsilon$. If the bias exceeds $(1-\frac{c_b}{\tau}) \epsilon$, we partially mitigate the bias using the postprocessing described in Section~\ref{Bias Analysis and Mitigation}, then the similar complexity analysis applies. We also extend this protocol to simultaneously estimate $M$ fidelities against $M$ different target states with sample complexity only increasing as $\mathcal{O}(\ln M)$.

\begin{theorem}[Fidelity Estimation for Most Quantum States]
\label{Theorem 2}
    Let $\vert \psi \rangle$ be an $n$-qubit pure state whose distribution $\pi$ satisfies the $k$-Generalized Local Escape Property defined in Definition~\ref{k-GLEP}, then the mixing time $\tau$ associated with $\vert \psi \rangle$ is bounded by $\mathcal{O}(\frac{n^2}{k}) $. Consequently, the fidelity against such states can be efficiently estimated with local Pauli measurements. 
\end{theorem}
We define the $k$-Generalized Local Escape Property ($k$-GLEP) in Section~\ref{Conditions for Fast Mixing}, which is a generalization of the local escape property from~\cite{10756060}. We prove in Appendix~\ref{path congestion} and~\ref{Multi-commodity Flow Analysis} that any target states $\vert \psi \rangle$ that satisfy $k$-GLEP have the mixing time polynomially bounded. $k$-GLEP is not restricted to any specific state ensembles; it is a set of conditions on the target state’s probability distribution in general. 
With Theorem~\ref{Theorem 2}, we can identify various classes of states, including $\epsilon$-approximate state $t$-designs, states prepared by certain low-depth random quantum circuits, and other physically relevant states like the ground state of (non-stoquastic) gapped local Hamiltonian, W states, and Dicke states, that have polynomially bounded mixing time, see Appendix~\ref{Haar random proof} to~\ref{W states} for detailed proofs. Most states should satisfy $k$-GLEP, hence admit efficient fidelity estimation.

\begin{result}[Empirical Check for Applicability (informal)]
\label{result:empirical-check}
There is a randomized test that, using $\textnormal{poly}(n^k, \frac{1}{\epsilon}, \ln \frac{1}{\delta}) $ classical queries to the target state's probability in the computational basis, will output \textnormal{\textsc{Pass}} with probability at least $1-\delta$ whenever the target state’s probability distribution satisfies the $k$-Generalized Local Escape Property on at least $1-\epsilon$ of its probability mass, and output \textnormal{\textsc{Fail}} otherwise.
\end{result}
We adopt the query access as Huang, Preskill, and Soleimanifar~\cite{10756060}: we are given query access to the amplitudes of the known target state $\vert \psi\rangle$, denoted as $\Psi(x) = \langle x | \psi \rangle$ for all $x \in \{0,1\}^n$, we assume the precision of the queried amplitudes is reasonably good to carry out the protocol faithfully, small amplitudes that are insignificant can be neglected. This empirical check only requires access to the probabilities $\vert \Psi(x) \vert^2$, which can be obtained from the same query access to amplitudes. We use rejection sampling to generate samples that represent the distribution $\pi$ \cite{devroye2006nonuniform, von195various}. 

This empirical check certifies whether the $k$-GLEP holds on most (at least $ 1-\epsilon$) of the probability mass; it does not guarantee $k$-GLEP holds pointwise for all basis states. It serves as a practical applicability test for the protocol. We establish its sample complexity using Hoeffding’s inequality and a union bound. One caveat is that the $k$-GLEP definition (see  Definition~\ref{k-GLEP} in Section~\ref{Conditions for Fast Mixing}) involves the support size of $\vert \psi \rangle$, which is not necessarily known for the state being checked. This quantity can be estimated via a collision estimator (reverse birthday paradox)~\cite{4557157}, which yields an effective support size sufficient for the tests. The procedure of this empirical check is listed in Protocol~\ref{Prot: Empirical Check} in Section~\ref {Main contect: Empirical Check}. Full details and proofs are presented in Appendix~\ref{Empirical check}.

\begin{theorem}[Mixed States Certification]
\label{Mixed States theorem}
    Let $\rho$ be an arbitrary state and let $\sigma = \sum_{i=1}^M p_i \vert \psi_i \rangle \langle \psi_i \vert, \ \ p_i \geq 0, \sum_{i=1}^M p_i=1, 1 \leq i \leq M $ be a mixed target, where each $\vert \psi_i \rangle$ is a pure state satisfying $k$-GLEP. Define $f_i =\langle \psi_i \vert \rho \vert \psi_i \rangle$. The fidelity between the arbitrary unknown state $\rho$ and the target mixed state is bounded as 
    \begin{equation*}
        \Big ( \sum_{i=1}^M p_i \sqrt{f_i}  \Big )^2 \leq F(\rho, \sigma) \leq \Big (\sum_{i=1}^M \sqrt{p_i f_i} \Big )^2
    \end{equation*}
\end{theorem}
Prior works have shown that we cannot efficiently certify or reconstruct arbitrary mixed states in general~\cite{Buadescu2019Certification}. A common response is to restrict to low-rank mixed states~\cite{chen2022toward, 2a6eac0a7fff47098b29d37402b877f0, gilyén2022improvedquantumalgorithmsfidelity, Cerezo_2020}, which offer a more tractable structure. However, the operational relevance of this assumption is questionable. Although low-rankness is widely used and not unreasonable, it is unclear what additional operationally meaningful insight this assumption provides over simply assuming purity. Such assumptions are sometimes introduced to serve theoretical convenience rather than to reflect experimental realities or to help us extract meaningful knowledge about quantum systems. 

Theorem~\ref{Mixed States theorem} provides an alternative and more operationally meaningful approach: we can efficiently determine whether the fidelity with respect to a mixed target state is above or below a threshold, thereby enabling certification of such mixed states without low-rank assumptions. Theorem~\ref{Main theorem} extends the protocol to reuse the same measurement data for fidelity estimation against $M$ different target states, with sample complexity only increasing as $\mathcal{O}(\ln M)$. This allows Theorem~\ref{Mixed States theorem} to provide the bounds efficiently. Detailed proofs are provided in Section~\ref{Mixed states}.

\section{Fidelity Estimation and Performance Guarantees}
\label{Main content}
\subsection{Fidelity Estimation Procedure}
\label{Procedure}
Our fidelity estimation protocol is detailed in Protocol~\ref{prot:fidelity} below:

\begin{protocol}[Estimating Fidelity $\langle \psi \vert \rho \vert \psi \rangle $ with $k$-local Pauli measurements]\label{prot:fidelity}
\textbf{Inputs:} 
$T$ samples of an unknown state $\rho$, query access $\Psi$ to the amplitudes of target state $\vert \psi \rangle$, integer $k \in \{1, \ldots, n \} $.
\sbline

\textbf{Output:} A point estimate $\hat{F}$ of fidelity $ F= \langle \psi \vert \rho \vert \psi \rangle $ 
\sbline

\textbf{Procedure:}
\begin{enumerate}
  \item \textbf{Measurement phase (on the unknown state $\rho$)}
  \begin{enumerate}
    \item
    Randomly choose a subset $A \subseteq [n]$ of $k$ qubits to be measured in random Pauli bases $ \{ X, Y, Z \}$, and $A^c$ is the subset of the remaining $n-k$ qubits

    \item
    Measure each qubit in $A^c$ in the $Z$ basis, yielding a bitstring $z \in \{0, 1\}^{n-k} $ 
    
    \item
    Measure each qubit in $A$ in a random Pauli basis (chosen independently per qubit), yielding outcome $s \in \{0, 1\}^k$
  \end{enumerate}
  \item \textbf{Query phase (using query access $\Psi$ of the target state $\vert \psi \rangle$)}
  \begin{enumerate}
    \item
    For each $2^k$-bitstring $x_k \in \{0, 1\}^k $, insert it into the positions in $A$, while fixing the other positions with bits from $z$, forming $x_z = x_k \,\Vert\, z$, $x_z \in \{0, 1\}^n $

    \item
    Query $\Psi(x_z)$ for all $x_k \in \{0, 1\}^k $

    \item
    Construct the (unnormalized) reduced state for $k$ qubit: $\vert \Psi_{A, z}' \rangle=\displaystyle\sum_{x_k \in \{0, 1 \}^k } \Psi (x_z) \vert x_k \rangle  $ and normalize it $\vert \Psi_{A, z} \rangle =\frac{\vert \Psi_{A, z}' \rangle}{\Vert \vert \Psi_{A, z}' \rangle \Vert} $

  \end{enumerate}
  \item \textbf{One local overlap estimation}
  \begin{enumerate}
    \item[] Use the single-qubit classical shadow formula on each qubit in $A$ and calculate the local overlap:
    \begin{equation*}
        \omega= \langle \Psi_{A, z} \vert \big( \bigotimes_{i\in A}(3 \vert s_i \rangle \langle s_i \vert -I) \big)  \vert \Psi_{A, z} \rangle
    \end{equation*}
    This yields one biased sample $\omega_t$ of the fidelity.
  \end{enumerate}
  \item \textbf{Final Estimation}
  \begin{enumerate}
  \item Repeat steps 1-3 $T$ times to obtain $\omega_1, \dots, \omega_T$, split them into $K$ disjoint batches of equal size $T_b= \lfloor T/K \rfloor$, for each batch, compute the mean: $\hat{\omega}_{b_j} = \frac{1}{T_b} \sum_{t=(j-1)T_b+1}^{jT_b} \omega_t $.
  \item Output $\hat{F}= \operatorname{median} \{\hat{\omega}_{b_1}, \ldots, \hat{\omega}_{b_K}  \} $.
  \end{enumerate}
\end{enumerate}

\sbline
\end{protocol}

\paragraph{Procedure.} We adopt the same assumptions as Huang, Preskill, and Soleimanifar~\cite{10756060}: for a known pure state 
\begin{equation}
\label{target state}
    \vert \psi \rangle = \sum_{x \in \{0, 1\}^n } \sqrt{\pi(x)} e^{i \phi(x)} \vert x \rangle,   \ \ \pi (x) := \vert \langle x \vert \psi \rangle \vert^2
\end{equation}
we have query access to the amplitudes of $\vert \psi \rangle$ in computational basis: $\Psi(x) =  \langle x \vert \psi \rangle $, $x \in \{0, 1\}^n $. And we have access to $T$ independent copies of an unknown state $\rho$. We follow the same data collection routine to obtain local overlap $\omega$ (steps 1-3 are the same with Protocol 2 in \cite{10756060}), but use different postprocessing (median-of-means), also alternative interpretation and analysis: the protocol directly outputs a point estimate $\hat{F}$ rather than a binary certification test; we analyze $\hat{F}$ as a (controlled-bias) fidelity estimator. Moreover, we extend the scope of this protocol substantially beyond Haar random states.

For each copy of $\rho$, we measure a randomly selected subset of $k$ qubits in randomly chosen Pauli bases, while measuring the remaining $n-k$ qubits in the computational ($Z$) basis. The measurement outcomes of the $n-k$ qubits on the computational basis, together with the query model, enable the construction of the $k$-qubit reduced conditional state of the target state $\vert \psi\rangle$. And the measurement of the $k$ qubits on a random Pauli basis enables the construction of a $k$-qubit classical shadow. 

We clarify some subtle points about the "Query phase", step 2 in Protocol~\ref{prot:fidelity}. The purpose of the query phase is to reconstruct the $k$-qubit reduced conditional state of the target state $ \vert \psi \rangle$, conditioned on the outcome $z \in \{0, 1\}^{n-k} $ of the remaining $n-k$ qubits in $A^c$. This conditional state is used to estimate local overlap, which is eventually used for the estimation of fidelity. The $2^k$ bitstrings $x_k \in \{0, 1\}^k$ correspond to basis states on the selected $k$-qubit subsystem $A$. For each $x_k$, we construct the full $n$-qubit bitstring $x_z \in \{0, 1\}^n $ by inserting $x_k$ into the position of $A$ and inserting $z$ into the position of $A^c$. We then use the query model to obtain each amplitude $\Psi(x_z)= \langle x_z \vert \psi \rangle$, with which we construct the reduced conditional state $\vert \Psi_{A, z} \rangle$ ($\vert \Psi'_{A, z} \rangle$). Note that a small constant $k$ (typically 1 or 2) is sufficient for many state classes we analyze (detailed proof in Section~\ref{Bounding the Mixing Time of the Markov Chain}); we only use $2^k$ queries per round. The string $x_k$ appears only in the model query phase, and the strings $s$ are used to obtain the classical shadow for the local overlap estimation.

If we need to estimate fidelity against $M$ pure target states $\{ \vert \psi_i \rangle \}_{i=1}^M $, with the query access $\Psi_i$ for each $\vert \psi_i \rangle $. The measurement phase remains the same, i.e., the experimental data collected on $\rho$ can be reused to estimate fidelities against multiple different target states. We run the (purely classical) query and postprocessing pipeline separately for each $\psi_i$, producing different sets of $ \{ \omega_t \} $, hence different fidelity estimates for each $\vert \psi_i \rangle $. The extension to simultaneously estimate multiple fidelities helps us extend the protocol to mixed state certification, i.e., Theorem~\ref{Mixed States theorem}, and assist in further extension, like tomography.

\paragraph{Analysis.} When measuring the $n-k$ qubit in $Z$ basis, each measurement outcome $z \in \{0, 1 \}^{n-k} $ selects a local neighborhood of strings $x$ differing by at most $k$ bits, thus probes a local slice of the target distribution $\pi$. To analyze how many measurements are needed for these local probes to collectively explore $\pi$ well enough for accurate fidelity estimation, we model the sequence of measurements as a reversible random walk on $\{0, 1\}^n $ that flips at most $k$ bits each step and has stationary distribution $\pi$: each measurement corresponds to one $k$-bit flip step. In this view, ``how many measurements are needed'' becomes ``how many $k$-bits–flip steps are required for the walk to converge to the stationary distribution $\pi$'', which is governed by the mixing time $\tau$.

Specifically, we define a Markov chain on the $n$-bit hypercube. Consider a weighted graph $G=(V, E_k)$ where $V= \{0,1 \}^n$ and two vertices $x, y \in V$ are connected by an edge if their Hamming distance satisfies $d(x,y) \leq k$. Each vertex $ x $ has a stationary weight $\pi(x) = \vert \langle x \vert \psi \rangle\vert^2$.
We define a Markov chain on $G$ with transition probabilities~\cite{10756060}:
\begin{equation}
\label{transition simple}
    P(x, y) \coloneqq \begin{cases}
        \displaystyle\frac{1}{N} \cdot \displaystyle\frac{\pi(y)}{\pi(x)+\pi(y)} & (x, y) \in E_k \\[0.35cm]
        1- \displaystyle\sum_{y \in V, y \neq x } P(x, y) & x=y \\[0.35cm]
        0 & \text{otherwise}
    \end{cases}
\end{equation} 
where $N=\sum_{r=1}^k \binom{n}{r} = \mathcal{O}(n^{k}) $ is the number of neighbors of a vertex.   

The transition matrix $P$ defines a reversible Markov chain with stationary distribution $\pi(x)$ induced by $\vert \psi \rangle$. Therefore, sampling from the stationary distribution of such an induced Markov chain is equivalent to sampling from $\pi(x) = \vert \langle x \vert \psi \rangle\vert^2$. The mixing time $\tau: = \frac{1}{1-\lambda_1}$ (i.e., inverse spectral gap of $P$) of a Markov chain quantifies how many steps are needed for the conditional sampling to adequately explore $\pi$. 
Essentially, Protocol~\ref{prot:fidelity} ``breaks down'' the global fidelity estimation into local overlaps while still capturing global correlations through the conditioning. The accuracy and efficiency of such ``break down'' can be analyzed with the mixing time $\tau$ of the Markov chain whose stationary distribution is $\pi $. We note that running Protocol~\ref{prot:fidelity} does not require the Markov chain; the Markov chain is only used as a mathematical tool in the analysis.

$\hat{F}$ is a biased estimator by design. Theoretically, classical shadow provides an unbiased estimator for fidelity. However, with the conditioned reduced target state $\vert \Psi'_{A, z} \rangle$, we inject bias in each sample of the local overlap, making them closer to the fidelity than the overlap between the target state and the original classical shadow, thus reducing the variance of the estimator. The empirical average of the local overlap is a biased estimator for fidelity, and it does not truly converge to the fidelity, however many samples we collect. However, for states that fall into the framework of Protocol~\ref{prot:fidelity} (and are not too far from the target states), the bias is small. Specifically, the upper bound of the bias is proportional to $\frac{1}{\tau}$. Increasing $k$ decreases the bias, as we will prove in Section~\ref{Multi-commodity Flow Analysis}; intuitively, increasing $k$ makes the protocol closer to classical shadow, which is an unbiased estimator for fidelity. We can also partially mitigate the bias with the knowledge of $\tau$, detailed in Section~\ref{Observables and Fidelity}. The mean square error (MSE) of an estimation consists of variance and bias. We prove with classical shadow formalism~\cite{Huang2020PredictingMP} and Hölder’s inequality that $\omega$ has bounded variance $\mathcal{O}(2^{2k})$. Standard concentration bounds imply $\Var(\hat{F})= \mathcal{O}(2^{2k}/T) $. We bound the bias to be $\frac{c_b}{\tau}$ fraction away from the total error $\epsilon$; with the triangle inequality, we take the bias into consideration in sample complexity analysis. Combining bias and variance gives
\begin{equation*}
    T= \mathcal{O} \Big ( \frac{2^{2k} \tau^2 }{(c_b \epsilon)^2} \ln \big (\frac{1}{\delta} \big ) \Big )
\end{equation*}
suffices to achieve additive accuracy $\epsilon$ with failure probability at most $\delta$. Thus, $\tau$ governs both accuracy and efficiency of the estimation. If we need to simultaneously estimate fidelity values against $M$ different pure target states, each with mixing time $\tau_i$, apply a union bound over all $M$ failure probabilities yields the sample complexity of $T = \mathcal{O} \Big (2^{2k} \frac{\max_i \tau_i^2}{(c_b \epsilon)^2} \ln \Big ( \frac{2M}{\delta} \Big ) \Big ) $. This concludes the overall analysis of Theorem~\ref{Main theorem}. Note that the median-of-means (MoM) postprocessing improves the concentration and robustness, but MoM does not mitigate the bias. We can also partially mitigate the bias. Detailed proofs of the sample complexity and the analysis of the bias and its mitigation methods can be found in Section~\ref{Observables and Fidelity}.

\subsection{\texorpdfstring{$k$}--Generalized Local Escape Property}
\label{Conditions for Fast Mixing}
This section establishes conditions on the probability distributions $\pi$ of the target state $\vert \psi \rangle$ that guarantee polynomially bounded mixing time for the induced Markov chain with stationary distribution $\pi(x)$. We formalize these conditions as the $k$-Generalized Local Escape Property ($k$-GLEP)~\cite{10756060}. Unlike assumptions tied to specific ensembles (e.g., Haar random), $k$-GLEP focuses on general structural properties of $\pi$: balanced weights and connectivity of the graph induced by $\pi$.

\begin{definition}[Good Vertex]
\label{Good Vertex}
    A vertex $x \in V $ is $(c_u, c_l)$-good if its stationary weight satisfies 
    \begin{equation*}
        \frac{c_l}{\vert \supp{(\pi)} \vert} \leq \pi(x) \leq \frac{c_u}{\vert \supp{(\pi)} \vert}
    \end{equation*}
\end{definition}
$\supp{(\pi)} = \{x: \pi(x) \neq 0 \} $ denotes the support of $\pi$.
\begin{definition}[$k$-Generalized Local Escape Property]
\label{k-GLEP}
    A state $\vert \psi \rangle=\sum_{x \in \{0, 1\}^n } \sqrt{\pi(x)} e^{i \phi(x)} \vert x \rangle $ \emph{satisfies $k$-GLEP} if its distribution $\pi$ (equivalently, the Markov chain induced by $\vert \psi\rangle$) obeys:   
    \begin{itemize}
    \item \textbf{Weight bound}: For all $x \in V$, $\pi(x) \leq \frac{c_u'}{\vert \supp{(\pi)} \vert}  $;
    \item \textbf{Smoothness}: Each $x \in V $ has at least $\alpha N$ neighbors that are $(c_u, c_l)$-good \footnote{The Smoothness condition requires $\vert \supp{(\pi)} \vert \geq \textnormal{poly}(n) $, though we do not explicitly list this requirement out. The Expansion condition also enforce $\vert \supp{(\pi)} \vert \geq \textnormal{poly}(n) $.};
    \item \textbf{Expansion}: For any two good vertices $x, y$ with Hamming distance $d(x,y) \leq 3k $, there exist at least $\alpha N$ pairwise internally disjoint paths of length at most $ 5 $  all of whose internal vertices are good.  \footnote{This Expansion condition can be relaxed. It is sufficient that such $\alpha N$ disjoint paths exist even if their internal vertices are not necessarily all good. In that case, rerouting through good neighbors preserves polynomial mixing bounds via multi-commodity flow analysis. Or it is also sufficient that there exists at least one path of length at most 5, all of whose internal vertices are good. Path congestion yields a polynomial bound in this case.} 
\end{itemize}
\end{definition}
Satisfaction of $k$-GLEP ensures polynomially bounded mixing time. 

\begin{lemma}[Connectivity]
\label{graph diam}
    Consider the reversible Markov chain defined in Equation~\ref{transition simple} on the $n$-bit hypercube graph $G=(V, E_k)$, whose stationary distribution is $\pi$. If this Markov chain satisfies $k$-GLEP,  then for any two vertices $x, y \in V$, there exist at least $\ceil{\frac{d(x, y)}{3k}} $ pairwise internally disjoint paths of length $ \mathcal{O} (\frac{n}{k} ) $ connecting $x$ and $y$ whose internal vertices are all good. 
\end{lemma}
We defer the proof of Lemma~\ref{graph diam} to Section~\ref{Overview and connectivity proof}; we sketch the intuition here. To construct the connecting paths and prove Lemma~\ref{graph diam}, we divide the bit positions where $x$ and $y$ differ into blocks of size at most $k$ and flip one block per step to build a path of length $\mathcal{O}(\frac{n}{k})$ connecting $x$ and $y$. $k$-GLEP ensures each step remains within well-weighted regions (smoothness), that alternative disjoint routes exist (expansion), and that no isolated heavy vertices impede the motion of the walk (weight bound). 

In section~\ref{Multi-commodity Flow Analysis}, we use multi-commodity flow to polynomially bound path congestion and prove that $k$-GLEP ensures the mixing time $\tau= \mathcal{O}(\frac{n^2}{k})$, i.e., a polynomial number of $k$-bit flip steps (equivalently, measurement rounds in Protocol~\ref{prot:fidelity}) suffice to explore the distribution $\pi$. Consequently, if the distribution of $\vert \psi \rangle$ satisfies $k$-GLEP, the fidelity against such states can be efficiently estimated with protocol~\ref{prot:fidelity} (Theorem~\ref{Theorem 2}). Then, in Section~\ref{Haar random proof} to~\ref{W states}, we prove Haar random states, state $t$-designs, states prepared by certain random low-depth quantum circuits, ground states of gapped local Hamiltonians, W states, and Dicke states all have polynomially bounded mixing time, thus fall into the scope of protocol~\ref{prot:fidelity}. By formulating the analysis in terms of this general $k$-GLEP condition rather than relying on ensemble-specific properties~\cite{10756060}, we broaden the scope of the protocol, covering more physically relevant states and states that can be prepared with noisy intermediate-scale quantum (NISQ) devices. These theoretical extensions are corroborated by our numerical results in Section~\ref{Result}.

\subsection{Empirical Check}
\label{Main contect: Empirical Check}
We have established that the $k$-GLEP guarantees polynomially bounded mixing time, hence efficient fidelity estimation; additionally, we have analytically proven several state ensembles satisfy $k$-GLEP. In practice, the target state may not belong to a known ensemble, or we cannot analytically verify if the target state falls into the scope of Protocol~\ref{prot:fidelity}. Therefore, we design an empirical check based on $k$-GLEP: check if $\pi$ satisfies $k$-GLEP with high probability using classical query access to $\pi$, which can be directly obtained from the query access $\Psi$ we assume for Protocol~\ref{prot:fidelity}. This empirical check provides a practical diagnostic tool for an arbitrary given state to test whether Protocol~\ref{prot:fidelity} can provide accurate fidelity estimates. Analysis of the sample complexity of this empirical check is included in Appendix~\ref{Empirical check}.

The empirical test first makes samples that represent $\pi$ using rejection sampling \cite{devroye2006nonuniform, von195various}, then uses these samples to estimate the support size using a collision estimator based on "reverse birthday paradox"~\cite{4557157}. Then we test all three conditions of $k$-GLEP for each sample. Together, they provide a certificate of $k$-GLEP on at least $1-\epsilon$ probability mass. Protocol~\ref{Prot: Empirical Check} lists the detailed empirical check procedure. 

\begin{protocol}[Check if $\pi:= \vert \langle x \vert \psi \rangle \vert^2 $ satisfy $k$-GLEP on at least $1-\epsilon$ probability mass]
\label{Prot: Empirical Check}
\textbf{Inputs:} 
Query access to $\pi$ (or query access $\Psi$ to the amplitudes of target state $\vert \psi \rangle$), $C$ for the ``Make samples'' step, $2^n\geq C \geq \sup_x \frac{\pi(x)}{2^{-n} } $. Accuracy and confidence parameters $\epsilon, \delta \in (0, 1)$, constants $c_u', c_l, c_u, \alpha$ and $k$ for $k$-GLEP. Sample size parameters $S \geq \frac{2}{\epsilon^2} \ln \frac{6}{\delta} $, $M \geq \frac{2}{\epsilon^2} \ln \frac{6 S}{\delta}$, $R \geq \frac{2}{\epsilon^2} \ln \frac{6S}{\delta}$.
\sbline

\textbf{Output:} \textsc{Pass}/\textsc{Fail} $k$-GLEP
\sbline
 
\textbf{Procedure:}
\begin{enumerate}
  \item \textbf{Make samples}
  \begin{itemize}
      \item[] Repeat: draw $x_U  \sim \textnormal{Unif} ( \{0, 1 \}^n)$, query $\pi(x_U)$, accept $x_U$ with probability $  \frac{\pi(x_U)}{C 2^{-n} } $. Stop until $S$ seed samples $x_1, \ldots, x_S  \sim \pi$ are accepted.
  \end{itemize}
  
  \item \textbf{Estimate $\vert \supp{\pi} \vert$ }
  \begin{itemize}
    \item If $\vert \supp{\pi} \vert$ is known, use the known $\vert \supp{\pi} \vert$;
    \item Estimate the effective support using collision estimator $\hat{C}_2$: Count the fraction of pairs $(i, j)$ with $x_i=x_j$, $\hat{C}_2=\frac{1}{S(S-1)  } \sum_{i \neq j} \mathbf{1} [ x_i=x_j ] $, the effective support size is $ \frac{1}{\hat{C}_2} $. If no repeats are observed, this collision estimator fails, set $\vert \supp{\pi} \vert \geq \binom{S}{2}/ \ln (\frac{1}{\delta}) $ or $\frac{2^n}{C}$.
  \end{itemize}
  \item \textbf{Conditions Testing}: For each sample $x_i$ check
  \begin{enumerate}
    \item \textbf{Weight bound}: Accept if $\pi(x_i) \leq \frac{c_u'}{\vert \supp{(\pi)} \vert}    $
    \item \textbf{Smoothness}: Sample $M$ i.i.d neighbors $y_1, \ldots, y_M$ uniformly from $\mathcal{N}_k(x_i)$, check $\pi(y_{i'})$, mark $y_{i'}$ as a good vertex if $\frac{c_l}{\vert \supp{(\pi)} \vert} \leq \pi(y) \leq \frac{c_u}{\vert \supp{(\pi)} \vert}$. Accept if at least $\alpha$ fraction of $M$ neighbors are good. 
    \item \textbf{Expansion}: Draw $R $ i.i.d. good pairs $(x, y)$ from the set $\{y: d(x_i, y) \leq 3k \} $ (reject and resample if $y$ is not good). For each pair, check if there are at least $\alpha N$ pairwise internally-disjoint paths of length at most 5 with all good vertices. Accept if at least a $(1-\epsilon) $ fraction of $R$ pairs admit such a path.
  \end{enumerate}
  \item \textbf{Decision}
  \begin{enumerate}
      \item[] Test (a), (b), (c) for each sample $x_i$, let $\hat{p}_{\rm weight} $, $\hat{p}_{\rm smooth} $, $\hat{p}_{\rm expand} $ be the acceptance fractions across all $S$ seed samples. Output \textsc{Pass} if $\hat{p}_{\rm weight} $, $\hat{p}_{\rm smooth} $, $\hat{p}_{\rm expand} $ are all at least $1-\frac{\epsilon}{2} $, otherwise, output \textsc{Fail}.  
  \end{enumerate}
\end{enumerate}
\sbline
\end{protocol}
We clarify the details of the steps in Protocol~\ref{Prot: Empirical Check} in Appendix~\ref{Empirical check}.

\section{Numerical Results}
\label{Result}
We conduct three series of numerical simulations to corroborate the theoretical results in Section~\ref{Performance guarantees} and to demonstrate some applications of our fidelity estimation protocol. These numerical experiments are designed to evaluate the protocol~\ref{prot:fidelity} from complementary perspectives:
\begin{itemize}
    \item \textbf{Benchmarking}: how well shadow overlap tracks the true fidelity compared to cross-entropy benchmarking (XEB) across states implementable on NISQ devices/physically relevant states under different noise models;
    \item \textbf{Optimization}: whether shadow fidelity can serve as a hardware-efficient cost function that mitigates barren plateaus induced by the global cost functions in variational quantum algorithms;
    \item \textbf{Physical relevance}: whether the protocol can reliably characterize and distinguish ground states of gapped Hamiltonians and classify different quantum phases of matter.    
\end{itemize}

\subsection{Benchmarking Quantum Devices}
Fidelity is a standard metric to quantify how closely an experimentally prepared $\rho$ to a known target state $\vert \psi \rangle$, and is widely used for benchmarking quantum devices. However, as quantum devices continue to scale up, the exponential resources needed for fidelity estimation pose a significant challenge. A common surrogate for large-scale systems is cross-entropy benchmarking (XEB), introduced for random-circuit sampling where ideal output probabilities approximately follow the Porter–Thomas distribution (i.e., an anticoncentrated distribution). In this regime and under a depolarizing channel model, XEB correlates with circuit fidelity and has a clean analytic relation~\cite{48651}. However outside this setting, e.g. structured states or non-depolarizing (e.g. dephasing) noises, XEB can deviate significantly from fidelity or even fail as a proxy. Recent analyses~\cite{ware2023sharpphasetransitionlinear, PRXQuantum.5.010334, Barak2020SpoofingLC, Morvan_2024} formalize when XEB tracks fidelity and document regimes where it fails or undergoes sharp breakdowns.

For a surrogate to be practical on NISQ devices and also informative, it should be: \begin{inparaenum}[1)]
    \item scalable in quantum and classical resources; 
    \item can be implemented with local, non-adaptive measurements, does not rely on sophisticated control; 
    \item tightly correlated with true fidelity; 
    \item applicable beyond restricted state ensembles or noise models; and 
    \item theoretically justified.
\end{inparaenum}
 Shadow fidelity estimated by Protocol~\ref{prot:fidelity} satisfies these criteria (see Section \ref{Performance guarantees}, Theorems \ref{Main theorem} and \ref{Theorem 2}).

We test our protocol against true fidelity and XEB on IQP states, Diagonal phase states Stabilizer (random Clifford) states, and the ground state of transverse-field Ising model (TFIM) under white, coherent, and dephasing noise: 
\begin{figure}[!ht]
\label{benchmark sim}
    \centering
    \includegraphics[width=\textwidth]{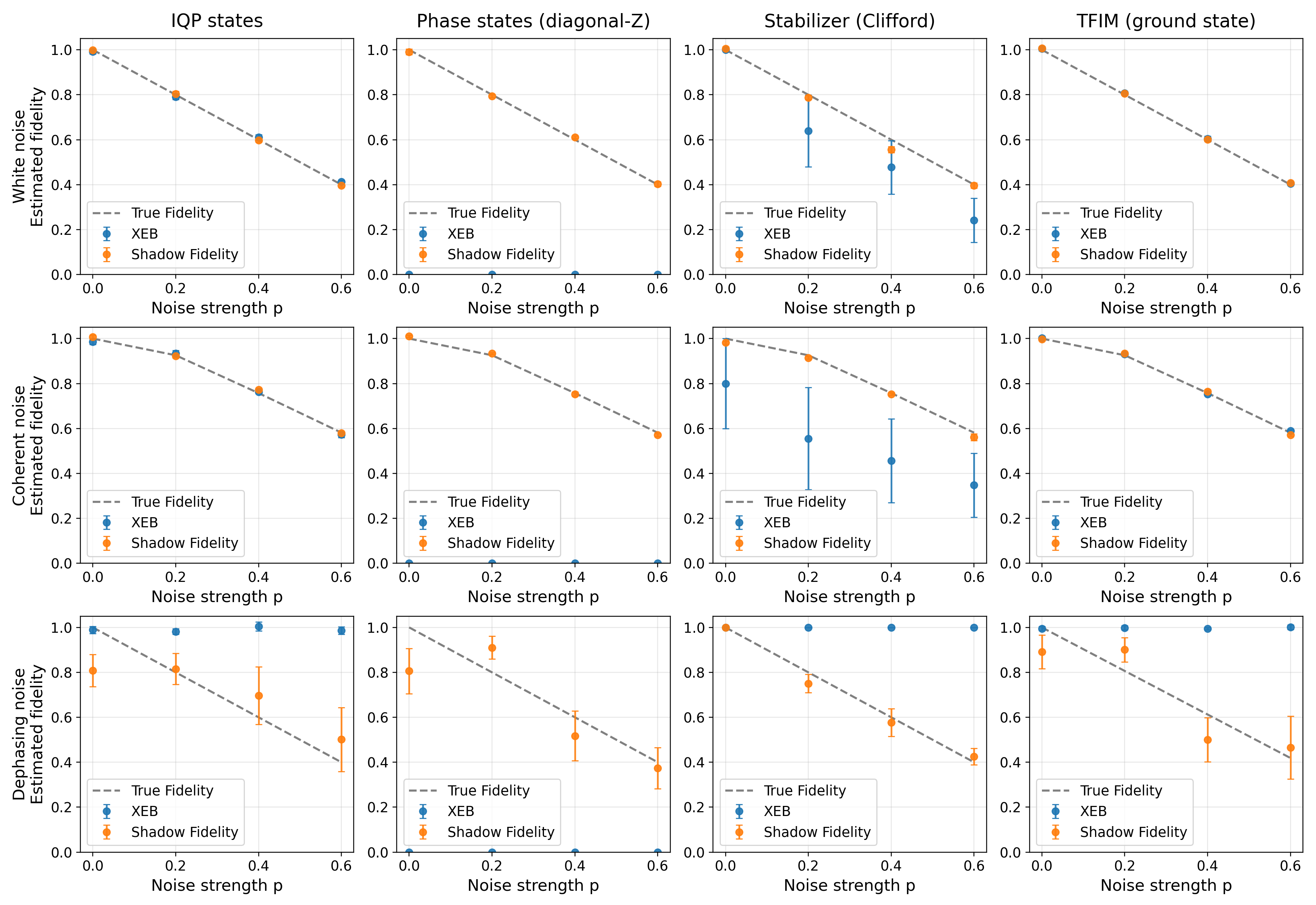}
    \caption{Performance comparison of XEB (blue), shadow fidelity (orange), and true fidelity (gray dashed) under white (top), coherent (middle), and dephasing (bottom) noise. Columns show IQP states, phase states (diagonal-$Z$), stabilizer (Clifford) states, and ground state of TFIM (TFIM). Shadow overlap closely follows true fidelity across all cases, while XEB fails for states that are not anticoncentrated and fails for dephasing noise.}
    \label{fig:benchmark_states}
\end{figure}

\begin{itemize}
    \item \textbf{IQP states}: $\vert \psi \rangle= H^{\otimes n} D_z H^{\otimes n} \vert 0^n \rangle$, where $D_z$ is a random diagonal-in-$Z$ unitary with phases of the form $\exp\big({\iu (\sum_i a_i x_i + \sum_{i<j }b_{ij}x_i x_j)}\big)$. Random IQP ensembles are known to be classically hard and anticoncentrated under natural choices. These are random low-depth circuits and have been proposed for quantum-advantage experiments~\cite{PhysRevLett.117.080501, Bremner_2017}. XEB is meaningful for IQP states. As shown in Figure~\ref{fig:benchmark_states}, both XEB and shadow overlap track fidelity well.
    \item \textbf{Diagonal phase states}: Dropping the final $H^{\otimes n}$ in the IQP circuit yields $\vert \psi \rangle= \frac{1}{\sqrt{2^n}} \sum_{x \in \{0, 1 \} } e^{\iu f(x)} \vert x \rangle $, where $f(x) = \sum_i a_i x_i + \sum_{i<j }b_{ij}x_i x_j $, i.e, uniform amplitude magnitude. Therefore, the ideal output distribution is uniform; XEB carries no information about fidelity for this type of states, which can be seen in Figure~\ref{fig:benchmark_states}. This type of states is implementable on current quantum devices, but XEB fails to characterize; in contrast, the shadow overlap tracks fidelity well.
    \item \textbf{Stabilizer (random Clifford) states}: States from random Clifford circuits and form 3-designs. They are highly structured and not anticoncentrated, thus XEB can be misleading, which is demonstrated in the third column of Figure~\ref{fig:benchmark_states}. However, the shadow overlap continues to perform well. We note that we extend Huang, Preskill, and Soleimanifar's work to state $t$-design in Section~\ref{t-design proof}; and this example corroborates our theoretical proof.
    \item \textbf{TFIM ground states}: $H_{TFIM}= -J \sum_i Z_i Z_{i+1} - h \sum_i X_i $, $J=1, h=1$ is a baseline model for correlated quantum many-body systems, and serves as a canonical benchmark in condensed matter and quantum simulation studies. As shown in Figure~\ref{fig:benchmark_states}, shadow fidelity accurately tracks the true fidelity under all noise models: this is one example that shadow fidelity reliably characterize physical relevant states.
\end{itemize}
Across all states, all noisy models, shadow fidelity consistently tracks the true fidelity. In particular, XEB fails to characterize the dephasing noise, since dephasing affects the phase coherence without changing measurement probabilities in the computational basis, making XEB insensitive to such state degradation, as shown in the last row of Figure~\ref{fig:benchmark_states}. In contrast, shadow fidelity can still characterize dephasing errors. In summary, these results also show that shadow overlap could serve as a scalable, experimentally practical, universal, and reliable proxy for fidelity for (at least part of) states prepared with NISQ devices and physically relevant states, even when XEB fails.

\subsection{Cost-Function-Dependent Barren Plateaus Mitigation}
Variational quantum algorithms (VQAs) have been considered as a practical application of NISQ devices; however, they are known to suffer from the barren plateau, where gradients vanish exponentially with system size, severely hurting trainability. One cause of the barren plateau is the global cost function. As studied by Cerezo \textit{et al.}~\cite{Cerezo_2021}, global cost functions, such as fidelity, exhibit exponentially vanishing gradients even for shallow circuits. By contrast, local cost functions only lead to polynomially vanishing gradients, thereby preserving trainability. Cerezo \textit{et.al} demonstrated this finding by replacing the global fidelity with a local fidelity in the context of quantum autoencoders (QAEs) and showed that the local fidelity mitigates barren plateaus and enables training at scales where the global fidelity fails. Such replacement is feasible for this example because the target state is $\vert 0 \rangle$. For a general state $\vert \psi \rangle = U \vert 0 \rangle^{n} $, measuring local fidelity involves $\frac{1}{n}\sum_{i=1}^n U ( \vert 0 \rangle \langle 0 \vert_i \otimes \mathbb{I}_{\Bar{j}} ) U^{\dagger} $, which is not necessary practical if the unitary $U$ is deep or complex. 

We adopt the same QAE setting introduced in Ref.~\cite{Cerezo_2021}. The system is partitioned into a retained subsystem $A$ ($n_A=1$) and a trash subsystem $B$ ($n_B$ qubits; we set $n_B=40$). The training ensemble consists of two product states $\vert \psi_1 \rangle = \vert 0 \rangle_A \otimes \vert 0 \rangle_B^{\otimes n_B} $, $\vert \psi_1 \rangle = \vert 1 \rangle_A \otimes \vert 1, 1, 0, \ldots, 0 \rangle_B $, with respective probabilities $p_1=\frac{2}{3} $ and $p_2=\frac{1}{3} $. The goal of the QAE is to encode relevant information into subsystem $A$ while projecting the trash subsystem $B$ to $\vert 0 \rangle^{\otimes n_B}$. So, when all qubits in $B$ are in $\vert 0 \rangle$, the QAE is successfully trained. The ansatz is a shallow hardware-efficient circuit with depth $L=2$, consisting of alternating layers of single-qubit $R_y(\theta) $ rotations and nearest-neighbor $CZ$ entangling gates in a brickwork pattern. We simulate the circuit using Matrix Product State (MPS) representation. We consider three cost functions: 
\begin{itemize}
    \item global cost: $C'_{G}=1-\Pr_B (0^{n_B}) $, i.e., one minus the probability that all trash qubits are measured in $\vert 0 \rangle$;
    \item local cost: $C'_{L}=1-\frac{1}{n_B} \sum_{j=1}^{n_B} \Pr_j (0) $, i.e., one minus the average probability that the trash qubits are in $\vert 0 \rangle$;
    \item shadow cost: one minus the estimated shadow fidelity against $\vert 0 \rangle$ state. 
\end{itemize}
We use the Simultaneous Perturbation Stochastic Approximation (SPSA) optimizer, a widely used gradient-based method that is sensitive to vanishing gradients (i.e. barren plateau). We also adaptively increase the measurement shots to control the stochastic noise. We limited to 350 iterations, with parameters updated as standard SPSA $a/t^{0.602} $ and $c/t^{0.101} $.

The numerical results on QAE ($n_A=1$, $n_B=40$) comparing the three cost functions are shown in Figure~\ref{fig:QAE sim}. The global cost (red) remains at 1 throughout training, signaling the barren plateau. In contrast, both local (blue) and shadow (green) costs decrease steadily, indicating the successful optimization. The shadow cost closely follows the local cost, demonstrating that the shadow fidelity inherits the favorable mitigating barren plateau properties of local costs. Unlike local fidelity, which may require complex basis rotations depending on the target, shadow fidelity is hardware-friendly, requiring only local Pauli measurements for arbitrary target states that fall into its scope. That is, shadow fidelity offers an efficient and experimentally practical substitute for mitigating cost-function-induced barren plateau.

\begin{figure}[t]
    \centering
    \includegraphics[width=\textwidth]{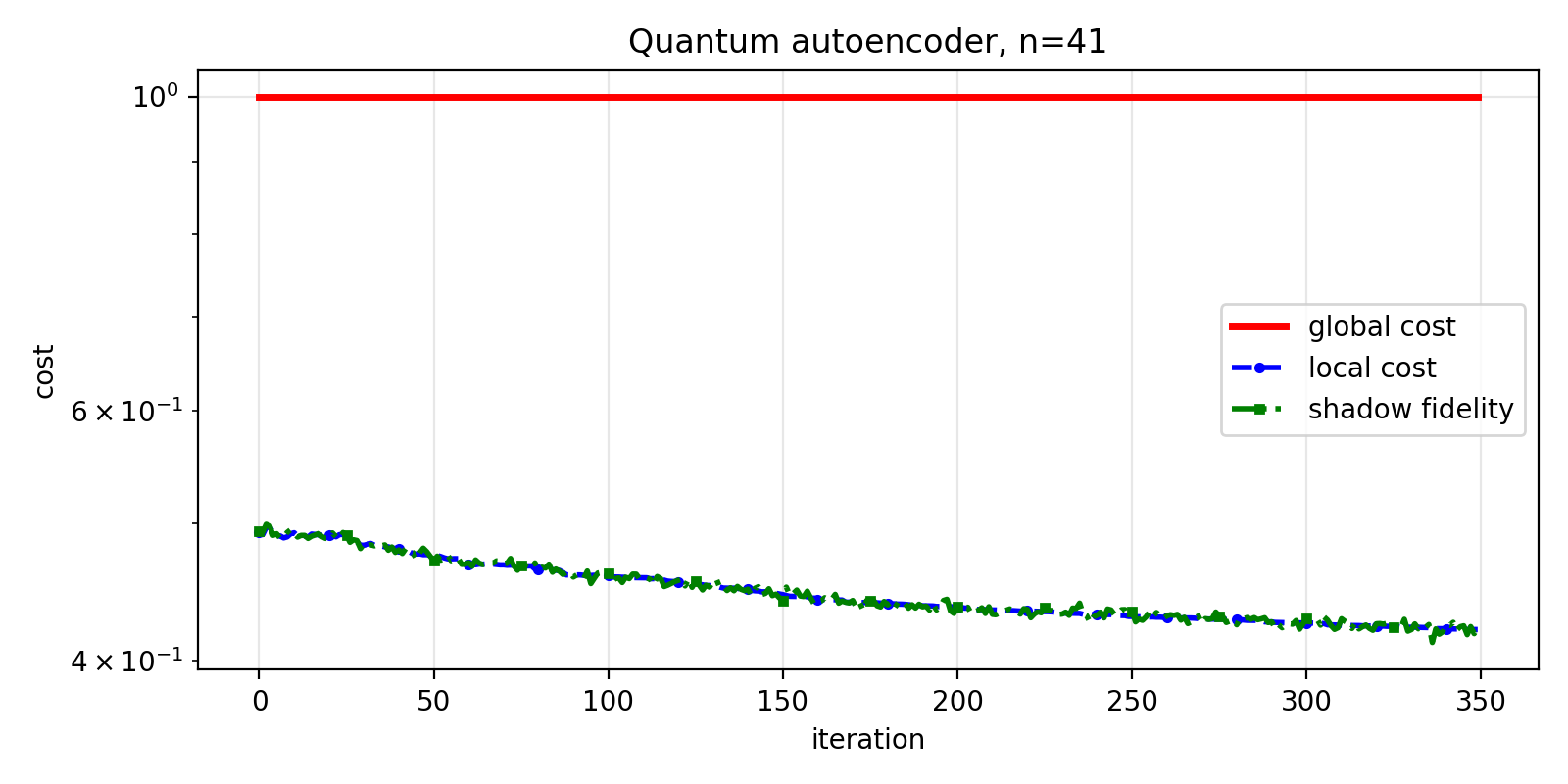}
    \caption{Training of a $41$-qubit quantum autoencoder ($n_A=1$, $n_B=40$) optimized with SPSA using: global cost (red), local cost (blue), and shadow fidelity cost (green). The global cost (red) remains flat near unity, indicating barren plateau, while the local (blue) and shadow fidelity (green) costs decrease steadily, indicating barren plateau mitigation. Shadow fidelity cost (green) closely tracks the performance of local cost (blue).}
    \label{fig:QAE sim}
\end{figure}

\subsection{Learning Properties of Ground States of Gapped Hamiltonians}
We now demonstrate that Protocol~\ref{prot:fidelity} can correctly estimate the fidelity for ground states of gapped local Hamiltonians across different phases. Apart from benchmarking (e.g., TFIM example in Figure~\ref{fig:benchmark_states}), accurate fidelity estimations against states in different phases enable phase classification even when the distinguishing observables are inherently global \cite{Huang_2022}.
We consider several representative models \cite{Huang_2022} acting on qubit systems with \emph{geometrically local} interactions. 
For geometrically local Hamiltonians, the interactions occur only between nearby sites on an underlying graph that fixes the qubits' spatial location.
We denote local Pauli operators acting on a site $i$ by $X_i, Y_i, Z_i$ and omit the tensor product symbol $\otimes$ as well as the tail of Identity operators when expressing a local term.

\begin{figure}[!ht]
    \centering
    \includegraphics[width=\textwidth]{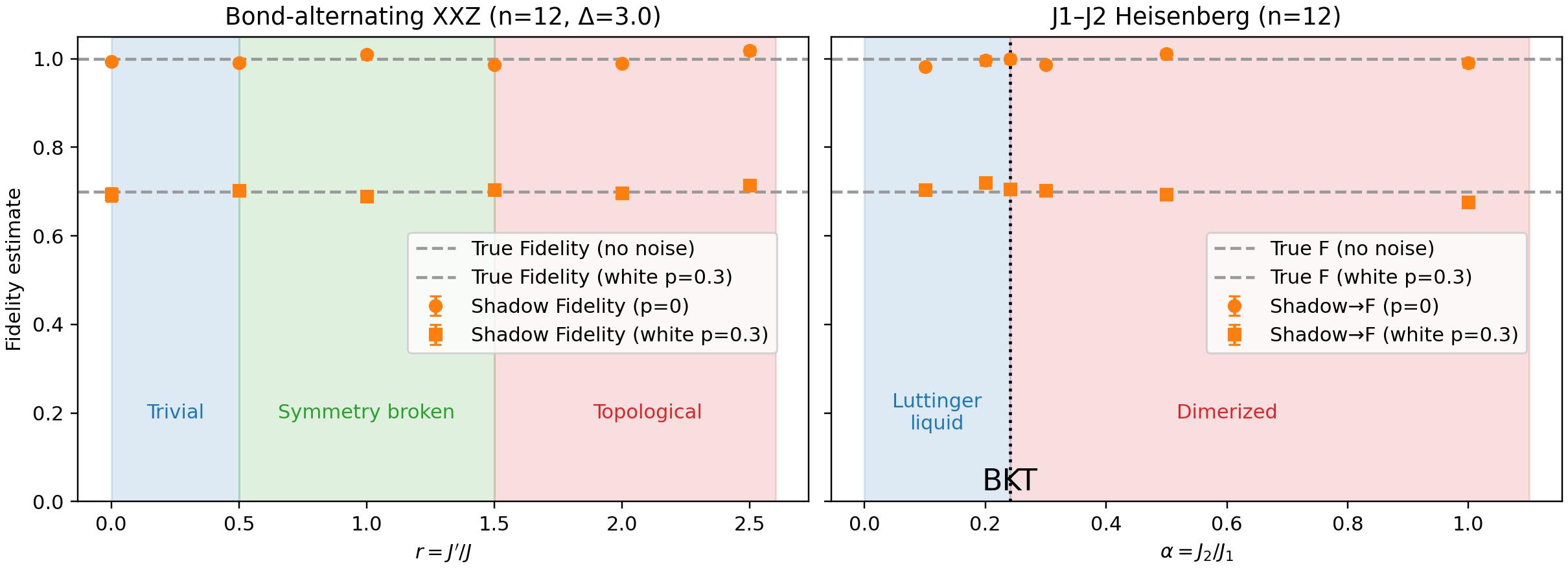}
    \caption{Estimated fidelities obtained from Protocol~\ref{prot:fidelity} for system size $n=12 $. The lab state is the ground state of $H_{bXXZ}$ (left)/$H_{bXXZ}$ (right). Left: Bond-alternating XXZ model with anisotropy $\Delta=3.0$, varying the ratio $r=J'/J$. The system transitions from a trivial phase (blue) to a symmetry-broken phase (green), and finally to a symmetry-protected topological (SPT) phase (red). Right: $J_1$–$J_2$ antiferromagnetic Heisenberg chain, varying $\alpha =J_2/J_1 $. A Berezinskii–Kosterlitz–Thouless (BKT) transition separates the gapless Luttinger liquid phase (blue) from the gapped dimerized phase (red) at the critical point $\alpha_c \approx 0.2411 $ (dashed line). For both Hamiltonians, the protocol accurately tracks the true fidelity (gray dashed lines) across phases, both under the noiseless case (orange circles) and the white noise with strength $p=0.3$ (orange squares). }
    \label{fig:bond_J1J2}
\end{figure}

We first consider the bond-alternating XXZ Hamiltonian,
\begin{equation*}
    H_{bXXZ}= \sum_{i: odd} J(X_i X_{i+1}+ Y_i Y_{i+1} + \Delta Z_i Z_{i+1} ) +  \sum_{i: even} J' (X_i X_{i+1}+ Y_i Y_{i+1} + \Delta Z_i Z_{i+1} )
\end{equation*}
where $J$ and $J'$ are alternating exchange couplings, and $\Delta$ controls the Ising anisotropy. This captures the physics of spin-Peierls materials, where lattice distortions generate alternating strong and weak interactions. As demonstrated in \cite{Huang_2022}, for $\Delta=3$, varying the ratio $r=J'/J$ drives the system through a sequence of distinct phases: a trivial phase at small $r$, then a symmetry-broken phase, and finally a symmetry-protected topological (SPT) phase at large $r$. This rich structure makes the bond-alternating XXZ Hamiltonian a compelling testbed for our fidelity estimation protocol. 

We also consider the frustrated $J_1$–$J_2$ antiferromagnetic Heisenberg chain
\begin{equation*}
    H_{J_1J_2}= \frac{J_1}{4} \sum_{i=1}^{n-1} (X_i X_{i+1}+ Y_i Y_{i+1} +  Z_i Z_{i+1} ) + \frac{J_2}{4} \sum_{i=1}^{n-2} (X_i X_{i+2}+ Y_i Y_{i+2} +  Z_i Z_{i+2} )
\end{equation*}
where $J_1, J_2 >0$. This Hamiltonian describes antiferromagnetic interactions between nearest ($J_1$) and next-nearest-neighbor ($J_2$) spins, introducing frustration: the competition between $J_1$ and $J_2$ prevents all bonds from being simultaneously minimized. Such Hamiltonians capture the physics of several quasi-1D magnetic compounds. The system exhibits a Berezinskii–Kosterlitz–Thouless (BKT) transition at the critical ratio $\alpha_c =J_2/J_1 \approx 0.2411$: for $\alpha < \alpha_c$, it is a gapless Luttinger liquid phase; for $\alpha > \alpha_c$, it is a gapped spontaneously dimerized phase. This transition does not involve a local order parameter, but the correlation length diverges; thus, conventional diagnostics are costly, and it provides a challenging test for Protocol~\ref{prot:fidelity}. At the Majumdar–Ghosh point $\alpha=0.5$~\cite{majumdar1969next}, we have a benchmark using the known structure of the ground states:
\begin{align*}
    \ket{\psi_+} &= \ket{\Psi^-}_{1,2} \otimes \ket{\Psi^-}_{3,4} \otimes \cdots \otimes \ket{\Psi^-}_{n-1,n}, & 
    \ket{\psi_-} &= \ket{\Psi^-}_{2,3} \otimes \ket{\Psi^-}_{4,5} \otimes \cdots \otimes \ket{\Psi^-}_{n,1},
\end{align*}
where $\ket{\Psi^-} = (\ket{01} - \ket{10})/\sqrt{2}$ is the singlet state.

We computed the ground states of both Hamiltonians for different system sizes while varying the parameters: $r= J'/J$ for $H_{bXXZ}$ and $\alpha=J_2/J_1$ for $H_{J_1J_2}$. This allows us to test different phases and transition points in each system. As shown in Figure~\ref{fig:bond_J1J2}, our protocol tracks the true fidelity across all phases, both in the noiseless case and under white noise with strength $p=0.3$. We note that in these examples the protocol remains robust across phase transitions, correctly identifying ground-state structure even at the BKT points. 

We next consider the 1D cluster-state parent Hamiltonian
\begin{equation}
\label{eq:1d-cl-Ham}
    H_{cl} = - \sum_{i=2}^{n-1} Z_{i-1} X_i Z_{i+1}.
\end{equation}
This non-stoquastic Hamiltonian has a unique ground state, i.e., the 1D cluster state, which serves as a resource state in measurement-based quantum computation (MBQC). Beyond its computational relevance, the 1D cluster state is in the SPT phase, protected by $Z_2 \times Z_2$ symmetry. 

We estimate fidelity with our protocol for different system sizes under both the noiseless conditions and in the presence of white noise with strength $p=0.3$. As shown in Figure~\ref{fig:cluster_states}, in both scenarios, the protocol estimates the fidelity accurately. We also generated five independent random trivial product states for each system size, constructed as tensor products of single-qubit states sampled uniformly from the Bloch sphere. These product states represent trivial phases without SPT order. Our protocol correctly estimates extremely low fidelity against these trivial phase product states, which sharply distinguishes between the two. This ability to discriminate topologically nontrivial resource states from trivial phases highlights the protocol’s capability of capturing global quantum order that is beyond the reach of few-body observables of the underlying ground state. All these results also corroborate our theoretical results that the ground states of gapped (non-stoquastic) Hamiltonians are within the scope of our protocol.

\begin{figure}[t]
    \centering
    \includegraphics[width=\textwidth]{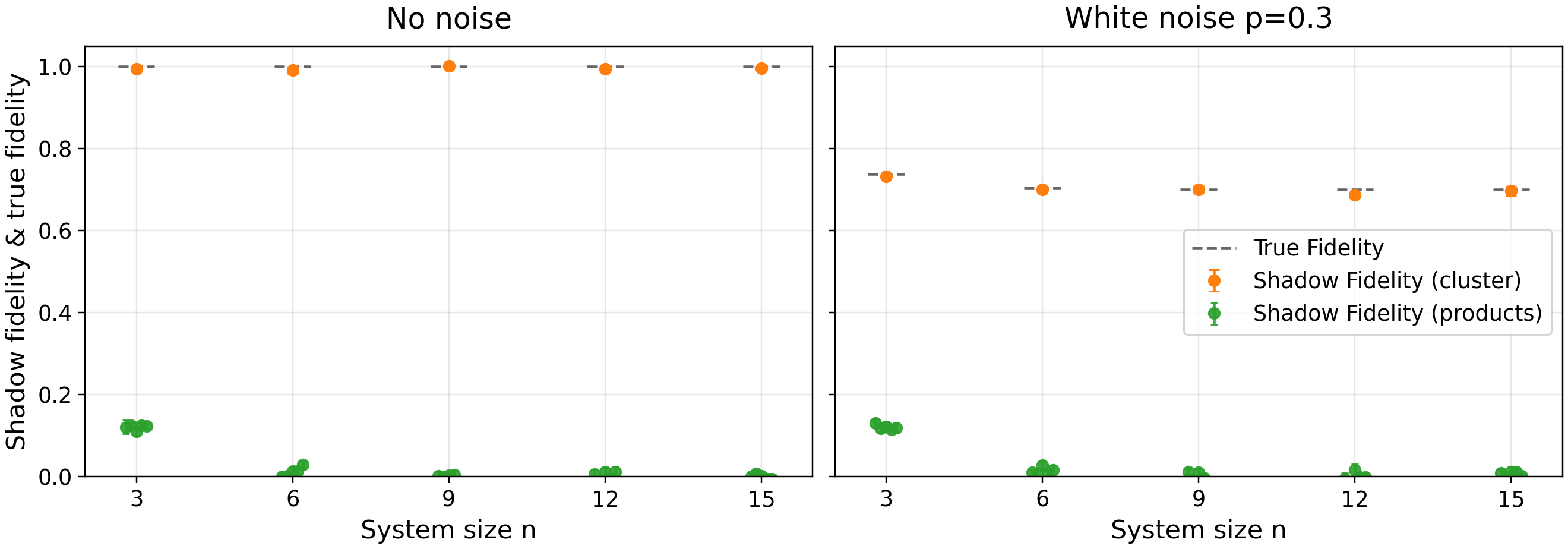}
    \caption{Estimated fidelities obtained from Protocol~\ref{prot:fidelity} for system sizes $n=3, 6, 9, 12, 15$. The lab state is the ground state of the cluster Hamiltonian (cf. Equation~\ref{eq:1d-cl-Ham}). Left: noiseless; right: white noise with strength $p=0.3$. Orange circles are the estimates against the true cluster ground state, which agree with the true fidelities (gray dashed lines). Green circles show fidelities estimated against five independent random product states, consistently yielding negligible values. The protocol correctly estimates the fidelity values and clearly distinguishes the highly entangled SPT cluster state from trivial product states.}
    \label{fig:cluster_states}
\end{figure}

In summary, our numerical results show that the fidelity estimation protocol accurately tracks the true fidelity across distinct quantum phases and remains robust at critical points and under noise. This demonstrates that the protocol captures essential global correlations beyond the reach of few-body observables and can serve as a scalable diagnostic tool for identifying and classifying quantum phases of matter, including those lacking conventional local order parameters.

\newpage
\section{Discussion}
\label{Discussion}
The local measurement schemes used in Protocol~\ref{prot:fidelity} are different from conventional $k$-local measurement schemes~\cite{Chen2014DiscontinuityOM, PhysRevLett.118.020401}; in each round, a random subset of $k$ qubits is measured in random Pauli bases, while the remaining $n-k$ qubits are also measured in the $Z$ basis, rather than left unmeasured. Such measurement schemes introduce no additional experimental overhead and provide crucial global information. Measuring the remaining $n-k$ qubits projects the system onto a definite configuration, collapsing the $k$ measured qubits into a conditional state. This conditional structure captures correlations between the $k$ randomly measured qubits and the remaining $n-k$ qubits, thus significantly reducing ambiguity when inferring global properties. In contrast, leaving qubits unmeasured is equivalent to tracing them out, which limits access to only the $k$-qubit reduced density matrix (RDM). As studied in~\cite{Chen2014DiscontinuityOM, PhysRevLett.118.020401}, $k$-RDMs can miss important global correlations. In particular, distinct $n$-qubit states may share the same collection of local marginals; such ambiguity is known as the quantum marginal problem. The additional $Z$ basis measurement results provide global information through conditioning, mitigating this ambiguity and enabling faithful fidelity estimation.

This fidelity estimation protocol does not impose restrictions on the target states $\vert \psi \rangle$ nor the unknown state $\rho$ and is broadly applicable. As we proved in Section~\ref{Bounding the Mixing Time of the Markov Chain}, many classes of states satisfy $k$-GLEP and thus are within the reach of our protocol. Given an arbitrary state outside these proven classes, we can also efficiently check if it satisfies $k$-GLEP.

Gupta, He, and O'Donnell introduced a certification protocol that can certify any quantum state using $\mathcal{O}(n^2)$ local measurements~\cite{gupta2025singlequbitmeasurementssufficecertify}. 
The key to allowing certification of any quantum state, instead of most quantum states, is the ability to adaptively learn and switch to a suitable basis tailored to the structure of the target state and the certification procedure. A typical example is the GHZ state: although the GHZ state does not satisfy $k$-GLEP in the computational basis (as it forms two isolated weight poles with no mixing), a basis change to the $X$ basis spreads amplitude support across the Hilbert space, thereby enabling certification. 
Gupta, He, and O'Donnell's protocol adaptively queries the probability in arbitrary product bases, and their adaptive calculation identifies the $X$ basis for the GHZ state.

Despite the generality to all states, our protocol has the following practical advantage:
\begin{itemize}
    \item The oracle access to probability in arbitrary product bases is strictly stronger than our query assumption. In general (worst-case scenario), simulating Gupta, He and O'Donnell's oracle from our query access takes exponentially cost: consider arbitrary tensor product basis $\Pi=\bigotimes_i\Pi_i  = \vert b\rangle \langle b \vert = \bigotimes_i\vert b_i\rangle \langle b_i \vert$, the oracle access require $\langle \psi \vert \Pi \vert \psi \rangle= \vert \langle b \vert \psi \rangle \vert^2$, where $\langle b \vert \psi \rangle=\sum_x \langle b \vert x \rangle \langle x \vert \psi \rangle  $ is a dense $2^n$-term summation in general. Such oracle access is used per measurement; for each state copy, there are $\mathcal{O}(n)$ oracle calls.
    \item Adaptive basis selection poses a challenge for practical implementation and limits its experimental use. For our protocol, the experimental part can be done without any particular query access to a certain state $\vert \psi \rangle$. Namely, an experimentalist can collect data from $\rho$ without knowing or querying $\vert \psi \rangle$, after completion of the data acquisition phase; they can conduct the query and the estimation against any $\vert \psi \rangle$. The experimental measurement phase does not depend on $\vert \psi \rangle$.
    \item Data cannot be shared/reused across multiple certification/estimation tasks with adaptive basis selection schemes. If we need to certify $\rho$ against $M$ different target states, the adaptive scheme requires distinct measurements per target. In contrast, our protocol supports data reuse: a single dataset allows fidelity estimation against $M$ target states, with sample complexity scaling only as $\mathcal{O} (\ln M)$. This advantage has very valuable applications: in this work, we use this advantage to extend our protocol to mixed states; in our following work, we use this in efficient tomography.
\end{itemize}

One natural variant of our protocol is measuring all $n$ qubits in independently chosen Pauli bases, rather than fixing the $n-k$ qubits to the $Z$ basis, as suggested in Ref.~\cite{li2024newgeneralquantumstate}. By allowing a more flexible basis, we possibly gain more information per measurement and improve the applicability of the protocol. However, the challenge lies in how to efficiently extract useful quantities (such as fidelity) from such data. The proposal in Ref.~\cite{li2024newgeneralquantumstate} constructs conditioned states based on full random Pauli measurements, but computing such conditional states is generally inefficient, and still relies on strong oracle access assumptions similar to Gupta, He, and O'Donnell's work~\cite{gupta2025singlequbitmeasurementssufficecertify}.

While our protocol does not cover all states, it is efficient, practical, and robust for a broad class of physically relevant states. Increasing $k$ (slightly) extends the range of states we can estimate, which is clear from our proof in Section~\ref{Bounding the Mixing Time of the Markov Chain}. Our protocol is not specific to restricted classes of states; it handles all states. When the protocol cannot give an efficient and accurate estimation for a specific state, we can know from an efficient check of $k$-GLEP as explained in Section~\ref{Empirical check}.

\vspace{0.5em}
\subsection*{Acknowledgments:}
The authors acknowledge Tuyen Nguyen for proofreading. MS acknowledges Hsin-Yuan Huang for insightful discussions during and after QIP 2025; acknowledges Ryan Mann for helpful guidance on the Markov chain; acknowledges Zhicheng Zhang for helpful discussions and suggesting related references. MS is supported by the UTS President's scholarship and the Sydney Quantum Academy Ph.D supplementary scholarship, and acknowledges the Australian Quantum Software Network
Microgrant sponsored by Hon Hai Research. GW is supported by a scholarship from the Sydney Quantum Academy and also supported by the ARC Centre of Excellence for Quantum Computation and Communication Technology (CQC2T), project number CE170100012. MB acknowledges the ARC Centre of Excellence for Quantum Computation and Communication Technology (CQC2T), project number CE170100012.

\subparagraph{Contributions.} 
MS developed the theoretical results, conducted the numerical experiments, and wrote the main part of the manuscript. 
GW assisted in proving the main results. All authors contributed ideas and assisted in preparing the manuscript. MB and CF supervised the project. 

\clearpage
\begin{appendices}

\section{Fidelity Estimation and Sample Complexity}
\label{Observables and Fidelity}
\subsection{Observables, Bias and Sample Complexity}
The fidelity between an unknown state $\rho$ and a known pure state $\vert \psi \rangle$ is defined as $F=\langle \psi \vert \rho \vert \psi \rangle$ and $\hat{F}$ denotes the fidelity estimate from protocol~\ref{prot:fidelity}. We adopt the analysis of Huang, Preskill, and Soleimanifar~\cite{10756060} (Appendix C, from equation 20 to the proof of $\Tr ( L \rho)= \mathbb{E }[\omega]$ above Theorem 4) that shows the observable used in Protocol~\ref{prot:fidelity} does act as an effective projector onto the target state $\vert \psi \rangle$. Below, we restate only the necessary definitions and conclusions from their proof for completeness. Our new analysis and derivation start after Equation~\ref{Shadow and Fidelity}.

Let $L$ denote the observable corresponding to the measurements in Protocol~\ref{prot:fidelity}; $S = \sum_{x \in V } \pi(x) \vert x \rangle \langle x \vert  $ be the diagonal scaling matrix defined by the target state probabilities $\pi(x) = \vert \langle x \vert \psi \rangle \vert^2$; and $\Phi=\sum_{x\in \{0,1 \}^n } e^{i \phi(x)} \vert x \rangle \langle x \vert $ be the diagonal phase matrix. Huang, Preskill, and Soleimanifar~\cite{10756060} showed that $L$ is related to the transition matrix $P$ (equation~\ref{transition simple}) of the corresponding Markov chain by
\begin{equation}
\label{L vs P}
    L=\Phi \cdot S^{\frac{1}{2}}P S^{-\frac{1}{2}} \cdot \Phi^{\dagger}
\end{equation}
Define $\vert \pi \rangle = \sum_{x \in \{0,1 \}^n } \sqrt{\pi(x)} \vert x \rangle$. $\sum_y P(x, y)= (1- \sum_{y \neq x } P(x, y)) + \sum_{y \neq x } P(x, y)=1$, so we have $S^{\frac{1}{2}}P S^{-\frac{1}{2}} \vert \pi \rangle  = \vert \psi  \rangle$. Applying the phase matrix gives $\Phi \vert \pi  \rangle= \vert \psi  \rangle$, therefore,
\begin{equation}
\label{L psi}
    \begin{split}
        L \vert \psi  \rangle &= \Phi \cdot S^{\frac{1}{2}}P S^{-\frac{1}{2}} \cdot \Phi^{\dagger} \vert \psi \rangle 
        = \Phi \cdot (S^{\frac{1}{2}}P S^{-\frac{1}{2}} )\cdot \Phi^{\dagger} \Phi \vert \pi \rangle) = 
        \Phi \vert \pi \rangle \\ &= \vert \psi \rangle
    \end{split}
\end{equation}
We order the eigenvalues of $L$ as $\lambda_0 \geq \lambda_1 \ldots \geq 0$. Equation~\ref{L psi} indicates $L \vert \psi \rangle = 1 \cdot \vert \psi \rangle$. Moreover, Huang, Preskill, and Soleimanifar proved that
\begin{equation}
\label{shadow 1}
    \Tr ( L \rho)= \mathbb{E }[\omega]
\end{equation}
As $0 \leq \mathbb{E }[\omega]= \Tr(L \rho) \leq 1$ for any state $\rho$, we have $0 \preceq L \preceq I $. So, $\lambda_0 = 1$ and $\vert \lambda_0 \rangle = \vert \psi \rangle$, then we can derive
\begin{equation}
\label{Shadow and Fidelity}
    \begin{split}
        \mathbb{E} [\omega] &= \Tr( L \rho )= \Tr( \sum_{i \geq 0} \lambda_i \vert \lambda_i\rangle \langle \lambda_i \vert  \rho ) = \sum_{i \geq 0} \lambda_i \Tr(\vert \lambda_i\rangle \langle \lambda_i \vert  \rho ) = \sum_{i \geq 0} \lambda_i  \langle \lambda_i \vert  \rho \vert \lambda_i\rangle \\ &= \langle \psi \vert \rho \vert \psi \rangle + \sum_{i \geq 1} \lambda_i \langle \lambda_i \vert \rho \vert \lambda_i \rangle\\
    \end{split}
\end{equation}
If we use $\frac{1}{T} \sum_{t=1}^T \omega_t$ as the estimator of the fidelity $F$, $\frac{1}{T} \sum_{t=1}^T \omega_t$ will concentrate to $\mathbb{E} [\omega]$, which is $F + \sum_{i \geq 1} \lambda_i \langle \lambda_i \vert \rho \vert \lambda_i \rangle$, not the fidelity $F$. That is, $\hat{F}=\frac{1}{T} \sum_{t=1}^T \omega_t$ as a fidelity estimator, it has a bias of $\sum_{i \geq 1} \lambda_i \langle \lambda_i \vert \rho \vert \lambda_i \rangle$. 
$\mathbb{E} [\omega]$ over estimate $F$ thus provides an upper bound of $F$. $\vert \lambda_i\rangle \langle \lambda_i \vert $ form an orthonormal basis, so $\sum_{i \geq 0} \langle \lambda_i \vert \rho \vert \lambda_i\rangle = \Tr(\rho) =1 $, $\sum_{i \geq 1} \langle \lambda_i \vert \rho \vert \lambda_i\rangle = 1- \langle \psi \vert \rho \vert \psi \rangle = 1-F$. And $\lambda_1 \geq \lambda_i, \forall i>1$, thus we can upper bound the bias term as 
\begin{equation}
\label{bias}
    Bias= \mathbb{E} [\omega] - F =\sum_{i \geq 1} \lambda_i \langle \lambda_i \vert \rho \vert \lambda_i \rangle \leq \sum_{i \geq 1} \lambda_1 \langle \lambda_i \vert \rho \vert \lambda_i \rangle  = \lambda_1 (1-F) \leq \lambda_1
\end{equation}
In the following, we first analyze the sample complexity of such a biased fidelity estimator, then propose methods to partially mitigate the bias.

As a convention of the analysis of mixing time of a Markov chain, $\frac{1}{\tau}$ is the spectrum gap between the largest and second largest eigenvalue, which is $\frac{1}{\tau} = 1 -\lambda_1$ in our analysis. With Equation~\ref{bias}, we have the upper bound for bias is
\begin{equation}
\label{bias 2}
    1-\frac{1}{\tau}
\end{equation}
We require that the bias is smaller than the total additive error $\epsilon$, that is, $\lambda_1 \leq \epsilon$, or $1-\frac{1}{\tau} \leq \epsilon$. For target states that fall into the scope of protocol~\ref{prot:fidelity}, as long as it is not too far from the unknown state $\rho$, the bias should not be too large, which is corroborated by the numerical simulations in Section~\ref{Result}. However, as we require $\epsilon$ to be larger than a non-zero bias, $\epsilon$ cannot be arbitrarily small. 

We use median-of-means as post-processing~\cite{Huang2020PredictingMP}, which is more robust to outliers and requires fewer samples than the empirical average used by Huang, Preskill, and Soleimanifar~\cite{10756060}. Accordingly, we analyze the sample complexity with the Chebyshev inequality below, rather than Hoeffding’s inequality used in Ref.~\cite{10756060}. We note that the key difference in the sample complexity analysis is that we take bias into consideration, which is achieved with the triangle inequality, and such a trick still applies even if we use the empirical average and analyze the complexity with Hoeffding’s inequality.
\begin{theorem}[Chebyshev inequality]
For random variable $X$ with finite expectation value $\mu$ and finite non-zero variance $\sigma^2$, for any real number $a>0$
\begin{equation*}
    \Pr ( \vert X - \mu \vert \geq a \sigma) \leq \frac{1}{a^2}
\end{equation*}
\end{theorem}
For each of the batches, apply Chebyshev inequality to shadow overlap $\omega$, set $\epsilon=a \sigma$, so $\frac{1}{a}=\frac{\sigma}{\epsilon}$, we have
\begin{equation}
\label{MoM each batch}
    \Pr \Big ( \Big \vert \frac{1}{T_b} \sum_{t=(j-1)T_b+1}^{jT_b} \omega_t  - \mathbb{E}_b[\omega] \Big \vert \geq \epsilon \Big ) \leq \frac{\sigma^2}{T_b \epsilon^2 }
\end{equation} 
for each batch, so $T_b = \mathcal{O}(\frac{\sigma^2}{\epsilon^2})$. Across the $K$ batches: at least half of the batch means must deviate by more than $\epsilon$ for the median to give a bad esitmation, Chernoff bound gives $\Pr ( \vert \frac{1}{T_b} \sum_{t=1}^{T_b} \omega_t - \mathbb{E}[\omega] \vert \geq \epsilon ) \leq 2 e ^{-K/2}  $. Set the failure probability $\delta= 2 e ^{-K/2}$, we get $K=2 \ln \frac{2}{\delta}$ batches. These analyze the sample complexity for estimating $\mathbb{E}[\omega]$, not the fidelity yet. 
For the estimation of fidelity, for each batch: $\frac{1}{T_b} \sum_{t=(j-1)T_b+1}^{jT_b} \omega_t - F = (\frac{1}{T_b} \sum_{t=(j-1)T_b+1}^{jT_b} \omega_t - \mathbb{E}_b[\omega] ) + (\mathbb{E}_b[\omega] -F)$, the triangle inequality gives $\Big \vert \frac{1}{T_b} \sum_{t=(j-1)T_b+1}^{jT_b} \omega_t - F \Big \vert \leq \Big \vert \frac{1}{T_b} \sum_{t=(j-1)T_b+1}^{jT_b} \omega_t - \mathbb{E}_b [\omega] \Big \vert  + \Big \vert \mathbb{E}_b [\omega] -F \Big \vert$, and $Bias=\Big \vert \mathbb{E}_b [\omega] -F \Big \vert$ so 
\begin{equation}
\label{MoM batch bias 1}
    \Pr \Big ( \Big \vert \frac{1}{T_b} \sum_{t=(j-1)T_b+1}^{jT_b} \omega_t  - F \Big \vert > \epsilon \Big ) \leq \Pr \Big ( \Big \vert \frac{1}{T_b} \sum_{t=(j-1)T_b+1}^{jT_b} \omega_t - \mathbb{E}_b [\omega] \Big \vert > \epsilon- \vert bias \vert  \Big ) \leq \frac{\sigma^2}{T_b (\epsilon- \vert bias \vert)^2 } 
\end{equation}
We now bound the variance $\sigma^2$, which is bounded by the range of $\omega = \Tr (L_{z_k} \hat{\rho}_k) $. Recall that $L_{z_k} := \vert \Psi_{A, z} \rangle \langle\Psi_{A, z} \vert$, $\vert \Psi_{A, z} \rangle =\frac{\sum_{x \in \{0,1 \}^k } \Psi(x_z) \vert x \rangle }{ \sqrt{\sum_{x \in \{0,1 \}^k } \vert \Psi(x_z) \vert^2}} $. Each $\Psi(x_z) \vert x \rangle$ contribute to a rank-1 projector, with the trace norm of 1. $\Vert L_{z_k} \Vert_1$ is bounded by the number of $\psi(x_z) \vert x \rangle$, which is $2^k$.  $\Vert \hat{\rho}_k \Vert_{\infty} \leq 2^k $, based on the analysis of classical shadow~\cite{Huang2020PredictingMP, 10756060}.
\begin{theorem}[Hölder’s inequality]
    For any two finite-dimensional matrices $A$, $B$ and any exponents $p, q \in[1, \infty] $ with $\frac{1}{p} +\frac{1}{q}=1 $
    \begin{equation*}
        \vert \Tr(AB) \vert \leq \Vert A \Vert_p \Vert B \Vert_q 
    \end{equation*}
\end{theorem}
Set $p=1$ (trace/Schatten-1 norm), $q= \infty$, we have 
\begin{equation}
\label{Hölder’s inequality for Schatten norms}
    \vert \Tr(AB) \vert \leq \Vert A \Vert_1 \Vert B \Vert_{\infty} 
\end{equation}
Apply Equation~\ref{Hölder’s inequality for Schatten norms} to $\vert \omega \vert$, we have 
\begin{equation}
\label{overlap variance}
    \vert \omega \vert = \vert \Tr (L_{z_k} \hat{\rho}_k) \vert \leq \Vert L_{z_k} \Vert_1 \Vert \hat{\rho}_k \Vert_{\infty} \leq 2^{k} \cdot 2^k =2^{2k}
\end{equation}
$\sigma^2$ is bounded by $2^{2k}$, insert this bound back to Equation~\ref{MoM batch bias 1}, and $K=2 \ln \frac{2}{\delta}$, we derive the sample complexity of 
\begin{equation}
\label{MoM bound 1}
    T \geq 2 \ln \frac{2}{\delta} \frac{2^{2k}}{(\epsilon-Bias)^2}
\end{equation}
$Bias$ is a variable, but we can bound this variable via $\tau$ according to Equation~\ref{bias} and ~\ref{bias 2}
\begin{equation*}
    (\epsilon- Bias)^2 \geq (\epsilon- \lambda_1)^2=(\epsilon- (1-\frac{1}{\tau}))^2
\end{equation*}
put the lower bound of $(\epsilon- Bias)^2 $ back into Equation~\ref{MoM bound 1}, we have $T \geq 2 \ln \frac{2}{\delta} \frac{2^{2k}}{(\epsilon- (1-\frac{1}{\tau}))^2}=2^{2k+1} \cdot \frac{\tau^2}{(1-(1-\epsilon)\tau)^2} \cdot \ln \frac{2}{\delta} $. To make this lower bound cleaner, we assume $\epsilon-(1-\frac{1}{\tau}) \geq \frac{c_b \epsilon}{\tau}$, where $\frac{c_b}{\tau} \in (0,1)$. This assumption requires the upper bound of the bias to be $\frac{c_b}{\tau} $ fraction away from the $\epsilon$, namely, the bias is not too close to $\epsilon$. \footnote{Intuitively, requiring the bias not to be almost as large as the total error $\epsilon$ is natural. We design such a factor with $\frac{1}{\tau}$ in the fraction rather than just $c_b$, which may not seem to be straightforward. However, this is a more reasonable design: the accuracy/bias is closely related to $\tau$.} Thus, we have $\frac{1}{(\epsilon-(1-\frac{1}{\tau}))^2} \leq \frac{\tau^2}{(c_b \epsilon)^2}$, so the lower bound becomes 
\begin{equation}
\label{complexity 2}
    T \geq 2^{2k+1} \cdot \frac{\tau^2}{(c_b \epsilon)^2} \cdot \ln \frac{2}{\delta} 
\end{equation}
If we need to simultaneously estimate $M$ fidelity, apply a union bound over all $M$ failure probabilities, then we need $K=2 \ln \frac{2M}{\delta}$ batches. Therefore, to guarantee an additive error $\epsilon$ on the $M$ fidelity estimations with failure probability at most $\delta$, we need
\begin{equation}
\label{complexity M}
    T \geq 2^{2k+1} \cdot \frac{\max_i \tau_i^2}{(c_b \epsilon)^2} \cdot \ln \frac{2M}{\delta} 
\end{equation}
This concludes the proof of Theorem~\ref{Main theorem} in section~\ref{Performance guarantees}.

\subsection{Bias Analysis and Mitigation}
\label{Bias Analysis and Mitigation}
The analysis above assumes the bias is within an acceptable range and sufficiently smaller than $\epsilon$. The bias is related to the spectral gap/mixing time and infidelity. As we have partial knowledge about the bias, we can: 1) provide a two-sided interval of the fidelity; 2) partially mitigate the bias.

As we discussed, $\mathbb{E}[\omega]$ provides an upper bound of the fidelity. We upper bound the bias in Equation~\ref{bias}, so
\begin{equation}
\label{deb_1}
    F \geq \hat{F}_{deb_1}= \mathbb{E} [\omega] - \lambda_1.
\end{equation} Put Equation~\ref{bias} back into Equation~\ref{Shadow and Fidelity} we have $\mathbb{E} [\omega] \leq F + \lambda_1 (1-F) $, with some reformulations, we have
\begin{equation}
\label{Shadow and Fidelity 2}
    F \geq \hat{F}_{deb_2} = \mathbb{E} [\omega] \frac{1}{1-\lambda_1}-\frac{\lambda_1}{1-\lambda_1}, 
\end{equation}
We use $F_{\min}= \max\{\hat{F}_{deb_1}, \hat{F}_{deb_2} \} $ as our lower bound. Substitute the $F_{\min}$ into Equation~\ref{bias}, the bias is upper bounded by $\lambda_1(1-F_{\min})$. If the bias is within an acceptable range, then any value in $[F_{\min}, \mathbb{E}[\omega]] $ can be used as an estimate of $F$. We suggest using the midpoint $\hat{F}_{mid}=\frac{1}{2} \mathbb{E}[\omega] + \frac{1}{2} F_{\min}$ of this interval, which is the Chebyshev center of this feasible set for $\hat{F}$. For $\hat{F}_{mid}$, the upper bound of bias is $Bias_{mid}= \min\{ \frac{\lambda_1}{2}, \frac{(1-\mathbb{E}[\omega])\lambda_1}{2(1-\lambda_1)} \} $.
As $Bias \leq \lambda_1(1-F)$, though $\mathbb{E} [\omega]$ overestimate $F$, it provides a data-assisted estimation
\begin{equation}
    \label{deb_3}
    \hat{F}_{deb_3}=\mathbb{E} [\omega] - \lambda_1 (1-\mathbb{E} [\omega])
\end{equation}
We summarize the five possible estimators and their residual biases
\begin{equation}
\label{deb summary}
\begin{aligned}
    \mathbb{E}[\omega] &  & \Delta_b&=Bias  \\
    \hat{F}_{deb_1}& =\mathbb{E}[\omega]-\lambda_1 & \Delta_1 &= \vert Bias-\lambda_1 \vert \\
    \hat{F}_{deb_2}& =\frac{\mathbb{E}[\omega]-\lambda_1}{1-\lambda_1} & \Delta_2 &= \Big \vert \frac{Bias-\lambda_1(1-F)}{1-\lambda_1} \Big \vert \\
    \hat{F}_{mid}&=\frac{1}{2} \mathbb{E}[\omega] + \frac{1}{2} F_{\min} &  \Delta_{mid}&= \min\{ \frac{\lambda_1}{2}, \frac{(1-\mathbb{E}[\omega])\lambda_1}{2(1-\lambda_1)} \}\\
    \hat{F}_{deb_3}& =\mathbb{E}[\omega]-\lambda_1(1-\mathbb{E}[\omega]) & \Delta_3 &= \vert (1+\lambda_1) Bias-\lambda_1 (1-F) \vert \\
\end{aligned}
\end{equation}
The last four estimators provide possible partial mitigation of the bias.
When $Bias$ is large, $\lambda_1$ and $1-F$ cannot be simultaneously small. We have rough knowledge about $\lambda_1$ (from $\tau$) and $F$ (from $[F_{\min}, \mathbb{E} [\omega]] $), which guide us on which estimator to use. We analyze three possible situations below.\\
\textbf{Case 1: $\lambda_1$ large, $F$ small}.\\
If the $Bias$ is large, this is the most likely situation: $Bias$ is upper bound by $\lambda_1 (1-F)$ and as we discussed in Section~\ref{Performance guarantees} and will further analyze in Section~\ref{Bounding the Mixing Time of the Markov Chain}, the efficacy and accuracy of this fidelity estimation protocol largely depend on $\frac{1}{\tau}=1-\lambda_1$. If $\lambda_1$ is large, we should check if $\tau$ is polynomially bound first; if not, our protocol is not applicable. Given that $\tau$ is still bounded, $\hat{F}_{deb_1}=\mathbb{E}[\omega]-\lambda_1$ provide the best estimator. As $\lambda_1$ is large, most of the Bias part is on the leading term $\lambda_1 \langle \lambda_i \vert  \rho \vert \lambda_i\rangle $; $F$ is small, so $1-F$ is not very far from 1. Therefore $Bias \approx \lambda_1 (1-F) \approx \lambda_1$. $\Delta_1$ is likely to be smaller than $\Delta_b$, $\Delta_2$ could be small because $Bias \approx \lambda_1 (1-F)$, but we need to consider the inflation due to $\frac{1}{1-\lambda_1}$. $\hat{F}_{deb_3}$ also provides a good estimate, though it may overestimate the fidelity; its residual bias is about $\lambda_1^2(1-F)$. $\hat{F}_{deb_3}$ has a larger variance constant, so $\hat{F}_{deb_1}$ is preferred. \\
\textbf{Case 2: $\lambda_1$ large, $F$ large}.\\
This may be the least likely situation, as we analyzed above. Given this circumstance, $\hat{F}_{deb_3}$ should be the best estimator, and $\mathbb{E}[\omega]$ is also a good estimator since the bias cannot be very large. As $(1-F)$ is far from 1, $\hat{F}_{deb_1}$ can severely overcorrect the bias. \\
\textbf{Case 3: $\lambda_1$ small, $F$ small}.\\
This situation happens when the protocol is applicable for the target state $\vert \psi \rangle$, but $\vert \psi \rangle$ is far from the unknown state $\rho$. Similar to Case 2, $\mathbb{E}[\omega]$ is still a good estimator since the bias cannot be very large. Another good candidate is $\hat{F}_{deb_1}$. Even though the fidelity estimator is biased, and the bias exceeds our acceptable range, this estimation still provides useful information. 

In practice, we choose the debiased estimator according to the analysis above, then estimate/upper-bound the residual bias according to the summary Equations~\ref{deb summary}. Check if the residual bias is within an acceptable range: if the new residual bias is acceptable, we adopt the corresponding estimator and the sample complexity. The analysis only changes slightly, but follows the same logic; if not, we need to further debias. Increasing $k$ (measuring more qubits in the random Pauli basis) changes $L$ ($P$), thus changes $\tau$, hence reduces the bias. The same logic for sample complexity analysis still applies. The sample complexity scales exponentially in $k$, so large $k$ is not practical due to scaling concerns. However, a slightly increasing i $k$ should help significantly, which will be clearer in Section~\ref{Bounding the Mixing Time of the Markov Chain}.  

Without increasing $k$, on the software level, we can power up $\omega$ and use extrapolation.
Consider Equation~\ref{Shadow and Fidelity} and
$L^k=\sum_{i \geq 0} \lambda_i^k \vert \lambda_i\rangle \langle \lambda_i \vert $
\begin{equation}
\label{Expectation power k}
\begin{split}
    \mathbb{E} [\omega^k]& =\Tr( L^k \rho )= \Tr (\sum_{i \geq 0} \lambda_i^k  \vert \lambda_i\rangle  \langle \lambda_i \vert  \rho) = \sum_{i \geq 0} \lambda_i^k  \langle \lambda_i \vert  \rho \vert \lambda_i\rangle = \lambda_0^k  \langle \lambda_0 \vert  \rho \vert \lambda_0\rangle +\sum_{i \geq 1} \lambda_i^k  \langle \lambda_i \vert  \rho \vert \lambda_i\rangle\\ & =  \langle \psi \vert  \rho \vert \psi \rangle +\sum_{i \geq 1} \lambda_i^k  \langle \lambda_i \vert  \rho \vert \lambda_i\rangle \leq F + \lambda_1^k (1-F)
\end{split}
\end{equation}
That is, if we take each of $\omega$ to the power of $k$, the expectation $\mathbb{E} [\omega^k]$ becomes $F$ plus a geometrically decayed bias $\sum_{i \geq 1} \lambda_i^k  \langle \lambda_i \vert  \rho \vert \lambda_i\rangle$~\cite{grier2024principaleigenstateclassicalshadows}. As $0 \leq \lambda_1 <1$ ($0 \leq \lambda_i <1$), ideally, if $k \to \infty$, the bias is gone. However, in practice, taking $k$ as large as possible is not the best strategy: an infinitely large $k$ is not practical; the relative error and numerical precision become non-trivial problems. In principle, we can estimate several $\mathbb{E} [\omega^k]$ with different $k$ and fit their decay to a model $F+c \lambda_{eff}^k $, which is a form of nonlinear regression whose accuracy hinges on a good signal-to-noise ratio at large $k$. For a practical choice of $k$, recall that we require the bias to be upper bounded by $\epsilon$, that is $\lambda_1^k(1-F) \leq \epsilon $, $(1-F)$ is at most 1, so we roughly take $\lambda_1^k \leq \epsilon$, thus $k=\ceil{\log_{\lambda_1} \epsilon } $. This leads to a power-moment estimator. Based on Equation~\ref{Expectation power k}, we use 
\begin{equation*}
    \hat{F}_{pow}= \mathbb{E} [\frac{\omega^k-\lambda_1^k}{1-\lambda_1^k} ]
\end{equation*}
as our estimator. Another way is to use Richardson extrapolation. The leading term of the bias is $\lambda_1^k \langle \lambda_1 \vert  \rho \vert \lambda_1\rangle$ ($k=1$ for $\mathbb{E} [\omega]$), this leading bias can be canceled via $\mathbb{E} [\omega^{k+1}] - \lambda_1 \mathbb{E} [\omega^k]$. The two-level Richardson extrapolation gives:
\begin{equation*}
    \hat{F}_{rich} =\mathbb{E} [\frac{\omega^{k+1} - \lambda_1 \omega^k}{1-\lambda_1}]
\end{equation*}
Note that if we use $\hat{F}_{rich} $, as the bias is reduced,
we can choose a smaller $k$ than $\ceil{\log_{\lambda_1} \epsilon }$ \footnote{We can roughly choose $\ceil{\frac{1}{2} \log_{\lambda_1} \epsilon}$, for a precise choice of $k$, we need to know $\lambda_2$.}, which save some shot counts and post processing cost. 

We need more analysis about this debias method in terms of how the variance or the range of $\omega^k$ changed. This depends on how we understand $\omega^k$ and $L^k$; we propose this debias method as a possible option to fully mitigate the bias, but theoretically, we need further rigorous analysis.
\section{Bounding the Mixing Time of the Markov Chain}
\label{Bounding the Mixing Time of the Markov Chain}
\subsection{Overview and Connectivity Proof}
\label{Overview and connectivity proof}
This section provides the detailed proofs for Theorem~\ref{Theorem 2}: if the probability distribution $\pi$ of the target state satisfies the $k$-GLEP, then the mixing time $\tau$ of the induced Markov chain is polynomially bounded. We follow the graph defined in Section~\ref{Procedure} and the transition probability defined in Equation~\ref{transition simple}. With satisfaction of the $k$-GLEP condition, we first bound the mixing time using path congestion in Section~\ref{path congestion}, then prove a tighter bound using multi-commodity flow in Section~\ref{Multi-commodity Flow Analysis}. We prove Haar random states, state $t$-designs, states prepared by certain random low-depth quantum circuits, ground states of gapped local Hamiltonians, as well as W states and Dicke states, all have polynomially bounded mixing time, from Section~\ref{Haar random proof} to Section~\ref{W states}. 

We now present the detailed proof of Lemma~\ref{graph diam} in Section~\ref{Conditions for Fast Mixing}, which establishes that under $k$-GLEP, any two vertices $x, y \in V$ are connected by at least $\ceil{\frac{d(x,y)}{3k}}$ pairwise internally disjoint good paths of length $\mathcal{O}(\frac{n}{k})$: 
\begin{proof}[Proof of Lemma~\ref{graph diam}]
    Let $F=\{ q \in [n]: x_q \neq y_q \} $ be the set of bit positions where $x$ and $y$ differ. Partition $F$ into $\ceil{\frac{d(x,y)}{k}} $ blocks and group these blocks into $g=\ceil{\frac{d(x,y)}{3k}} $ triples: $B_{1_1}, B_{1_2}, B_{1_3}, B_{2_1} B_{2_2}, B_{2_3}\ldots, B_{g_1}, B_{g_2}, B_{g_3} $. To construct the $i$-th disjoint path, we begin by flipping the bits in the first block $B_{i_1}$ of the $i$-th triple, changing the corresponding bits $v_i$ of $x$ to obtain a new vertex with $v_i'$. We then proceed by flipping the remaining blocks but never revisit $B_{i_1}$, so $v_i' \neq v_i$ throughout the path. For the $j$-th path with $j \neq i$, we begin with $B_{j_1}$ and do not flip $B_{i_1}$ initially, so $v_i' = v_i$ until we move to $i$-th triple, by which time at least another $3k$ bits are flipped. As a result, paths corresponding to different triples never share the same vertex; they are $3k$ bits away.
    This construction yields $\ceil{\frac{d(x,y)}{3k}} $ pairwise internally disjoint paths. Each path involves $\mathcal{O} (\frac{d(x, y)}{k}) \leq \mathcal{O} ( \frac{n}{k} )$ bit-flip steps, so the path length is bounded accordingly.
    
    To ensure that all internal vertices are good, the path detours through a good neighbor whenever it encounters a bad vertex $e^+$. 
    For any two vertices $(e^+, e^-) $, their corresponding good neighbors $g(e^+)$ and $g(e^-)$ satisfy $d(g(e^+), g(e^-)) = d(e^+, e^-)+ d(e^+, g(e^+))+ d(e^-, g(e^-)) \leq 3k $. Hence, there are at least $\alpha N$ internally disjoint good-only paths of length at most $ 5$ connecting $g(e^+)$ and $g(e^-)$. We substitute any bad vertex with a corresponding good-only subpath of length at most 5. Each substitution increases the length by at most $2+5=7$ (2 steps to enter/exit the neighbor, 5 for the good path),
    thus the total path length is bound by $\mathcal{O}(\frac{n}{k})$. Finally, since any two paths are separated by at least $3k$-bit differences, and detours only involve $k$-local neighborhoods, paths remain disjoint. 
\end{proof}
Lemma~\ref{graph diam} and its proof assist the following proofs in this section.

\subsection{Path Congestion Analysis}
\label{path congestion}
We now analyze the mixing time using the path congestion method.
Let $\Gamma =\{ \gamma_{xy}\} $ denote a collection of simple paths, each connecting a unique ordered pair of distinct vertices $(x, y)$, the path congestion parameter is defined as~\cite{10.5555/646385.689853, levin2009markov, 10756060}
\begin{equation*}
    \rho(\Gamma) = \max_{e\in E_k } \frac{1}{Q(e) } \sum_{\gamma_{xy} \ni e } \pi(x) \pi(y) \vert \gamma_{xy} \vert
\end{equation*}
where $Q(e)=\pi(e^+) P(e^+, e^-)$ is the weight (the probability of the occurrence of transition$(e^+, e^-)$). The spectral gap is bounded by the path congestion as $1-\lambda_1 \geq \frac{1}{\rho(\Gamma)}$, and the mixing time $\tau \leq \rho(\Gamma)(\ln \pi(x)_{\rm min}^{-1}+\ln \epsilon^{-1}) $, where $\pi(x)_{\rm min}$ is the minimum of $\pi(x)$, and $0 < \epsilon <1$~\cite{Sinclair_1992}. To upper bound the mixing time $\tau$, we need to upper bound the path congestion $\rho(\Gamma)$.

According to the transition rule defined in equations~\ref{transition simple}, we have: $Q(e)=\frac{1}{N} \frac{\pi(e^+)\pi(e^-)}{\pi(e^+)+ \pi(e^-)} $, therefore:
\begin{equation}
\label{weight}
\begin{split}
\frac{1}{Q(e)} & = N \frac{\pi(e^+)+ \pi(e^-)}{\pi(e^+)\pi(e^-)} = N \Big (\frac{1}{\pi(e^-)} + \frac{1}{\pi(e^+)} \Big ) = \frac{N}{\pi(e^+)} \Big (1+\frac{\pi(e^+)}{\pi(e^-)} \Big )\\
\end{split}
\end{equation}
Lemma~\ref{graph diam} ensures that the shortest path $\gamma_{xy}$ connecting every pair $(x, y)$ has length
\begin{equation}
\label{length}
    \vert \gamma_{xy} \vert \leq \mathcal{O} (\frac{n}{k})
\end{equation}
The summation over $\gamma_{xy} \ni e$ sums up all the ordered pairs $(x, y)$ such that $\gamma_{xy}$ includes edge $e$; the number of such pairs is at most $\mathcal{O}(\vert \supp{(\pi)} \vert)$ \footnote{May be tighten when $k >1$}. For worst-case analysis, we assume each term $\pi(x)\pi(y)$ is maximized, we have:
\begin{equation}
\label{sum edges}
    \sum_{\gamma_{xy} \ni e } \pi(x) \pi(y) \vert \gamma_{xy} \vert \leq \vert \supp{(\pi)} \vert \cdot \pi(x) \pi(y) \cdot \mathcal{O} (\frac{n}{k})
\end{equation}
Combining equations~\ref{weight} and~\ref{sum edges}, we have
\begin{equation*}
    \rho (\Gamma)  \leq \frac{N}{\pi(e^+)} \Big (1+\frac{\pi(e^+)}{\pi(e^-)} \Big ) \cdot \vert \supp{(\pi)} \vert \pi(x) \pi(y) \cdot \mathcal{O} (\frac{n}{k})  
\end{equation*}
The internal vertices are all good; we only need to consider three cases based on the location of good or bad vertices.\\
\textbf{Case 1: Both $x$ and $y$ are good}.\\
Both $e^+$ and $e^-$ are good vertices: $\frac{c_l}{\vert \supp{(\pi)} \vert} \leq \pi(e^+) \leq \frac{c_u}{\vert \supp{(\pi)} \vert} $, $\frac{c_l}{\vert \supp{(\pi)} \vert} \leq \pi(e^-) \leq \frac{c_u}{\vert \supp{(\pi)} \vert} $, we have $\frac{c_l  }{c_u } \leq\frac{\pi(e^+) }{\pi(e^-)} \leq \frac{c_u }{c_l }  $, therefore $\frac{1}{Q(e)} \leq \frac{N}{c_l 2^{-n}} (1+\frac{c_u}{c_l})$. The path congestion is bounded as:
\begin{equation}
\label{good-good}
\begin{split}
    \rho (\Gamma)  & \leq \frac{N \vert \supp{(\pi)} \vert}{c_l } (1+ \frac{c_u}{c_l} ) \cdot \vert \supp{(\pi)} \vert (\frac{c_u}{\vert \supp{(\pi)} \vert})^2 \cdot \mathcal{O} (\frac{n}{k}) \\
    & = \mathcal{O} (\frac{n^{k+1} }{k}) \frac{c_u^2}{c_l}  (1+ \frac{c_u}{c_l}) 
\end{split}  
\end{equation}
\textbf{Case 2: One of $x, y$ is a bad vertex}.\\
Without loss of generality, we assume $x$ is the bad vertex. The same analysis applies to the case of $y$ being a bad vertex. The main concerns is if $\pi(x) <  \frac{c_l}{\vert \supp{(\pi)} \vert}$, exponentially small $\pi(x)$ could blow up $\frac{1}{Q(e)}$, however the $\pi(x=e^+)$ in the summation actually cancel the one in $\frac{1}{Q(e)}$, as derived below:
\begin{equation}
\label{one bad}
\begin{split}
  \rho(\Gamma) & \leq \frac{N}{\pi(e^+)} \Big (1+\frac{\pi(e^+)}{\pi(e^-)} \Big ) \cdot \vert \supp{(\pi)} \vert \pi(x) \pi(y) \cdot \mathcal{O} (\frac{n}{k})\\ & = \frac{N}{\pi(e^+=x)} \Big (1+\frac{\pi(e^+=x)}{\pi(e^-)} \Big ) \cdot \vert \supp{(\pi)} \vert \pi(x=e^+) \pi(y) \cdot \mathcal{O} (\frac{n}{k}) \\ 
  &= N \Big (1+\frac{\pi(e^+=x)}{\pi(e^-)} \Big ) \cdot \vert \supp{(\pi)} \vert \pi(y) \mathcal{O} (\frac{n}{k}) \\
\end{split}  
\end{equation}
If $\pi(x) < \frac{c_l}{\vert \supp{(\pi)} \vert} $, $\Big (1+\frac{\pi(e^+=x)}{\pi(e^-)} \Big ) \approx 1$, Equation~\ref{one bad} is bounded as:
\begin{equation*}
    \rho (\Gamma)  \leq \mathcal{O} (\frac{n^{k+1} }{k}) c_u
\end{equation*}
If $ \frac{c_u}{\vert \supp{(\pi)} \vert} < \pi(x) \leq \frac{c_u'}{\vert \supp{(\pi)} \vert} $, $\Big (1+\frac{\pi(e^+=x)}{\pi(e^-)} \Big ) \leq (1 +\frac{c_u'}{c_l}) $, Equation~\ref{one bad} is bounded as \footnote{$c_u'$ could be $\Theta(n)$ (Haar random), then each $c_u'$ contributes an additional order of $n$ in the final bound}:
\begin{equation*}
    \rho (\Gamma)  \leq \mathcal{O} (\frac{n^{k+1} }{k}) c_u (1+ \frac{c_u'}{c_l})
\end{equation*}
\textbf{Case 3: Both $x$ and $y$ are bad}. \\
By construction, if a path starts at a bad vertex $x$, it immediately goes to a good neighbor, and Lemma~\ref{graph diam} guarantees that all internal vertices are good; thus, the vertex before $y$ is good. No edge $e = (e^+, e^-)$ along any such path has both endpoints as bad vertices. As a result, even though both $x$ and $y$ can be bad, any edge in the path lies between either a bad and a good vertex or two good vertices. Thus, this case reduces to the same analysis as \textbf{Case 2}.

This proves that the mixing time $\tau$ of any states that satisfy $k$-GLEP is polynomially bounded, partially proving Theorem~\ref{Theorem 2} in Section~\ref{Performance guarantees}.

\subsection{Multi-Commodity Flow Analysis}
\label{Multi-commodity Flow Analysis}
To obtain a tighter bound, we analyze with multi-commodity flow, where the demand $\pi(x)\pi(y)$ between each pair $(x, y)$ is distributed across multiple paths, rather than a single canonical path.
Let $\mathcal{P}_{xy} $ denote the set of simple directed paths connecting $x$ and $y$. A flow function $f$ assigns a nonnegative value to each path $p \in \mathcal{P}{xy}$ such that $\sum_{p \in \mathcal{P}_{xy}} f(p) = 1$. The demand $\pi(x)\pi(y) $ is split so that a fraction $f(p)$ travels along each path $p \in \mathcal{P }_{xy}  $. The resistance $R(f)$ of the flow $f$ is defined as
\begin{equation*}
    R(f) = \max_{e \in E_k}  \frac{1}{Q(e) } \sum_{x, y} \sum_{p \in \mathcal{P}_{xy}: p \ni e} \pi(x) \pi(y) f(p) \vert p \vert
\end{equation*}
where $Q(e)$ is the weight, $\frac{1}{Q(e) }$ is analyzed in Equation~\ref{weight} above. $\vert p \vert$ is the number of edges in path $p$ and is bounded by $\mathcal{O}(\frac{n}{k})$ due to Lemma~\ref{graph diam}. The spectral gap is bounded by flow as $1-\lambda_1 \geq \frac{1}{R(f)}$, and the mixing time $\tau \leq R(f)(\ln \pi(x)_{\rm min}^{-1}+\ln \epsilon^{-1}) $, where $\pi(x)_{\rm min}$ is the minimum of $\pi(x)$, and $0 < \epsilon <1$~\cite{Sinclair_1992}. The upper bound of the flow $R(f)$ indicates the upper bound of the mixing time $\tau$.

We now construct a layered flow structure. For a pair $(x, y)$ with Hamming distance $d(x, y)$, let $l =\ceil{\frac{d(x,y)}{k} } \leq \frac{n}{k}$ and the $m$-th layer is defined by:
\begin{equation*}
    L_m = \{ e^+ \in V: d(e^+, y)=d(x, y) - mk \}, \ \ 0 \leq m \leq l
\end{equation*}
$L_0=\{ x \} $ and $L_l=\{ y \} $. 
We construct the flow from $x$ to $y$ by uniformly distributing it through the good vertices at each layer. Each step of the flow moves one layer. $k$-GLEP guarantees that each $e^+$ has at least $\alpha N $ good neighbors. Among these, at most $\binom{d(e^+, y)}{1}+ \ldots + \binom{d(e^+, y)}{k-1} \leq \mathcal{O}( n^{k-1})$ are in the same layer, for the rest, they are either in the upper layer or the lower layer. So there are at least $\alpha' n^{k}$ legitimate edges ($\alpha' >0$). At each step, the flow is evenly split across these $\alpha' n^{k}$ edges. so the weight on each path $p$ used in the flow is bounded by $\frac{1}{\alpha' n^k}$.
There are at most $\mathcal{O}(\vert \supp{(\pi)} \vert)$ pairs $(x,y)$, for worst-case analysis, we assume each term $\pi(x)\pi(y)$ is maximized, we have:
\begin{equation}
\label{flow edges}
    \sum_{x, y} \sum_{p \in \mathcal{P}_{xy}: p \ni e} \pi(x) \pi(y) f(p) \vert p \vert \leq \vert \supp{(\pi)} \vert \cdot \pi(x) \pi(y) \frac{1}{\alpha' n^ {k}} \cdot \mathcal{O}( \frac{n}{k})
\end{equation}
Combining equations~\ref{weight} and~\ref{flow edges}, we bound the resistance as: 
\begin{equation*}
    R(f) \leq N (\frac{1}{\pi(e^-)} + \frac{1}{\pi(e^+)}) \cdot \vert \supp{(\pi)} \vert \pi(x) \pi(y) \frac{1}{\alpha' n^ {k}} \cdot \mathcal{O}( \frac{n}{k})
\end{equation*}
We analyze this in three cases:\\
\textbf{Case 1: Both $x$ and $y$ are good}.
$\pi(x), \pi(y) \in \Big [\frac{c_l}{\vert \supp{(\pi)} \vert}, \frac{c_u}{\vert \supp{(\pi)} \vert} \Big ] $, so $\pi(e^+), \pi(e^-) \in \Big [\frac{c_l}{\vert \supp{(\pi)} \vert}, \frac{c_u}{\vert \supp{(\pi)} \vert} \Big ] $:
\begin{equation}
\begin{split}
    R(f) & \leq \frac{N \vert \supp{(\pi)} \vert}{c_l } (1+\frac{c_u}{c_l} ) \vert \supp{(\pi)} \vert ( \frac{c_u}{\vert \supp{(\pi)} \vert} )^2 \frac{1}{\alpha' n^ {k}} \cdot \mathcal{O}( \frac{n}{k}) \\
    & = \mathcal{O}( \frac{n}{k}) \frac{c_u^2 }{c_l \alpha'} (1+\frac{c_u}{c_l} )
\end{split}
\end{equation}
\textbf{Case 2: One of $x, y$ is a bad vertex}.\\
With a similar analysis to that in Section~\ref{path congestion}, if $x < c_l 2^{-n} $:
\begin{equation}
\begin{split}
    R(f) & \leq N \cdot \vert \supp{(\pi)} \vert \pi(y) \frac{1}{\alpha' n^ {k}} \cdot \mathcal{O}( \frac{n}{k}) \\
    & = \mathcal{O}( \frac{n}{k}) \frac{c_u  }{ \alpha'} 
\end{split}
\end{equation}
If $c_u 2^{-n} <x \leq c'_u 2^{-n} $:
\begin{equation}
\begin{split}
    R(f) & \leq  \mathcal{O} (\frac{n }{k}) \frac{c_u}{\alpha'}  (1+ \frac{c'_u}{c_l}) 
\end{split}
\end{equation}
\textbf{Case 3: Both $x$ and $y$ are bad}. \\
With a similar analysis to that in Section~\ref{path congestion}, this case reduces to \textbf{Case 2} above.\\

This proves that the mixing time $\tau$ of any states that satisfy $k$-GLEP is bounded by $\mathcal{\mathcal{O}}(\frac{n^2}{k})$, which proves Theorem~\ref{Theorem 2} in Section~\ref{Performance guarantees}.

Multi-commodity flow analysis permits a relaxation of the expansion condition in $k$-GLEP: internal vertices of paths can be bad. Whenever a flow path encounters a bad vertex, the flow is split/rerouted among its good neighbors. This avoids exponentially small values of $\pi(e^+)$ or $\pi(e^-)$ in $\frac{1}{Q(e)}$.
However, the edges with both good endpoints carry additional flow that would have otherwise passed through bad vertices, and we need to consider this factor. Each bad vertex has at least $\alpha N$ good neighbor, each good neighbor receives at most $\frac{1}{\alpha N}$ additional flow; each good vertex has at most $N$ neighbor, so the additional flow due to bad vertices is bounded by $\frac{1}{\alpha N}N = \mathcal{O}(1)$. Thus, the edges $e$ whose endpoints are good vertices carry at most $\mathcal{O}(1)\cdot \mathcal{O}(1) = \mathcal{O}(1)$ additional flow due to bad vertices. This constant overhead does not affect the resistance bounds derived above, and the mixing time remains the same scaling.

\subsection{Haar Random State}
\label{Haar random proof}
Haar measure formalizes the notion of choosing unitaries uniformly at random~\cite{Haar_intro}. Haar random states $\vert \psi \rangle= U \vert 0\rangle^{\otimes n} $, $U \sim \text{Haar}$, are the canonical model for ``typical'' pure states. Porter-Thomas distribution describes the outcome probabilities when such a state is measured in the computational basis. 
For the induced distribution $\pi(x) = \vert \langle x \vert \psi  \rangle \vert^2 $ where $\vert \psi  \rangle $ is Haar random state, we model $\pi(x)$ as follows: for each $x \in \{0, 1\} $, independently sample $z(x)$ from the exponential distribution $\Pr(z)=e^{-z} \cdot \mathbbm{1} [ z \geq 0 ] $, and set $\pi(x) = \frac{z(x)}{2^n} $. We prove that such $\pi(x)$ satisfies $k$-GLEP with high probability. For Haar random state, $\vert \supp{(\pi)} \vert= 2^n$, we directly replace $\vert \supp{(\pi)} \vert$ in Definition~\ref{k-GLEP} and~\ref{Good Vertex} with $2^n$.
\begin{prop}
\label{prop Haar}
For $z_1, \ldots, z_m$ drawn independently from the exponential distribution $\Pr(z)=e^{-z} \cdot \mathbbm{1} [ z \geq 0 ]$, we have~\cite{10756060}
\begin{equation*}
    \Pr(z_k \leq t)= 1-e^{-t} 
\end{equation*}
\end{prop}
To prove the upper bound for all the weights, with proposition ~\ref{prop Haar}, we have the failure probability: $\Pr(z_k > t)= e^{-t} $. Let $D=2^n$, set $t=c \ln D$, we have $\Pr[z(x) > c \ln D  ] = D^{-c}$, that is:
\begin{equation*}
    \Pr[\pi(x) > \frac{c \ln D }{D} ] = D^{-c}
\end{equation*}
Apply union bound over all $x \in \{0, 1\}^n $ yields the failure probability of $D^{1-c}$, which is exponentially small for $c \geq 2$. Namely, with $t= c \ln 2^n = c \cdot n \ln 2  $ and $c \geq 2$,  $\Pr(\pi(x) \leq t 2^{-n})= 1-e^{-t} $, the failure probability is exponentially small. Therefore, for $c_u' =\Theta (n) $, \textbf{Weight bound} $\pi(x) \leq c_u' 2^{-n} $ holds with high probability. 

We now establish proof for the \textbf{Smoothness} condition. The Smoothness condition: $c_l 2^{-n} \leq \pi(x) \leq c_u 2^{-n} $ requires $c_l \leq z(x) \leq c_u  $. Let $p= \Pr[c_l \leq z(x) \leq c_u] $ be the probability that a vertex is good. (For the Smoothness condition to be satisfied with high probability, we need a reasonably high $p$. In this section, we follow the choices in Ref.~\cite{10756060} and set $c_l=1/11, c_u=5$. In general, we can derive the requirements of $c_l, c_u$, and we include such derivation in Section~\ref{t-design proof}.) Each vertex has $N$ neighbors, each neighbor is a good neighbor with probability $p$, and one a neighbor is good or not is independent of another neighbor. Namely, the event that a given neighbor is good is an independent Bernoulli variable with success probability $p$.
\begin{theorem}[Hoeffding's inequality]
\label{Hoeffding’s inequality}
    Let $x_1, \ldots, x_N$ be $N$ independent random variables taking values in $[a, b]$, let $X = \frac{1}{N} \sum_{i=1}^N x_i $. For any $t>0$,
    \begin{equation*}
        \Pr(\vert X - \mathbb{E}[X] \vert  \geq t) \leq 2 e^{-2Nt^2/(b-a)^2}
    \end{equation*}
\end{theorem}
With Theorem~\ref{Hoeffding’s inequality}, we have
\begin{corollary}
\label{tail bound Bernoulli}
    Let $x_1, \ldots, x_N$ be $N$ i.i.d. Bernoulli random variables with $\Pr[x_i=1]=p $ for $i \in [N]$. 
    \begin{equation*}
        \Pr \Big( \sum_{i=1}^N x_i \leq \alpha N \Big) < 2^{-\frac{ 2(p-\alpha)^2N}{\ln 2}}
    \end{equation*}
\end{corollary}
\begin{proof}
    For i.i.d. Bernoulli random variables $x_1, \ldots, x_N$ with $\Pr[x_i=1]=p $ for $i \in [N]$, $\mu = \mathbb{E}[X=\sum_{i=1}^N x_i ] = Np $, $b=1, a=0$. Let $t= Np-\alpha N$, $\mu-t =Np-Np + \alpha N=\alpha N$, with Hoeffding's inequality, we have $\Pr \Big( X \leq \alpha  N \Big) < e^{- \frac{2(p-\alpha)^2N^2 }{N}}= e^{ -2(p-\alpha)^2N } $. We then convert the base $e$ into 2: $2^b=e^{b \ln 2} $, set $e^x=2^b$, then $e^x=e^{b \ln 2}$, $b=\frac{x}{\ln2}$, so $e^x=2^{\frac{x}{\ln2}} $. $e^{ -2(p-\alpha)^2N} =2^{-\frac{ 2(p-\alpha)^2N}{\ln 2}} $ 
\end{proof}
We apply corollary~\ref{tail bound Bernoulli}, so for each vertex,
\begin{equation*}
    \Pr( \# \text{good neighbors} \leq \alpha  N ) < 2^{-\frac{ 2(p-\alpha)^2N}{\ln 2}}
\end{equation*}
Apply union bound over all $2^n$ vertex yield the failure probability of \textbf{Smoothness} less than $ 2^{n-\frac{ 2(p-\alpha)^2N}{\ln 2}} $. $N=n^k$, when $k=1$, $ 2^{(1-\frac{ 2(p-\alpha)^2}{\ln 2})n} $ is not necessarily very small; however, once $k>1$, $N$ is much larger than $n$, the failure probability of \textbf{Smoothness} is exponentially small in $n$. Thus, the \textbf{smoothness} condition holds with high probability.

We now address the \textbf{Expansion} condition. When $ \pi(x)$ is sampled from the Porter-Thomas distribution as defined above, the induced graph $G$ over the hypercube $\{0,1\}^n$ has full support; that is, all vertices have nonzero weight and are connected via the Markov chain.
\begin{lemma}
\label{Haar random 5 path}
    For two vertices $x, y$ on the $n$-bit hypercube whose Hamming distance $d(x, y) \leq 3k$, there exist $N$ internally disjoint paths connecting them with length at most $5$.
\end{lemma}

\begin{proof}
    The construction of the path follows a similar gist to the proof of Lemma~\ref{graph diam}. We first consider the case $k=1$. For every $i \in [n]$, we write bit position cyclically as $(i, i+1, i+2, \ldots, n, 1, \ldots, i-1) $. We construct a path $\gamma_i$ in the following way: flip the $i$th bit of $x$ to get $x'$, then scan cyclically through all remaining bits ($i+1 \rightarrow n \rightarrow 1 \rightarrow i-1$) and flip those where $x'$ and $y$ differ. As $d(x, y) \leq 3$, $d(x', y) \leq 4$, there are at most 4 flips during the scanning. So the total length of such a path is bounded by $5$. During the scan, every bit position is visited at most once; hence, the vertex sequence has no repetitions, and each path is simple. For any internal vertex $v_i$ of $\gamma_i$, its $i$-th bit satisfies $v_i =x'_i \neq x_i$. For a different path $\gamma_j$, $v_i=x_i$ until the scan moves to bit $i$ by which time other bits are flipped. So an interior vertex cannot belong to both $\gamma_i$ and $\gamma_j $ unless $i=j$, that is, the interiors are pairwise disjoint. We illustrate one such path construction in Figure~\ref{Draw a 3/4 bit hypercube, x=000, y=111, draw all the path with different colors}.

    When $k > 1$, the bit flips needed to reach $y$ can be grouped into at most 3 blocks, each block has size at most $k$. There are $N$ different ways for the first flip. We follow the same gist as the case $k=1$: for each of the $N$ ways to start the first flip, we construct a corresponding path and correct any residual mismatch in the final step. The same disjointness argument applies.
\end{proof}

\begin{figure}[!ht]
    \centering
    \begin{tikzpicture}
    		[cube/.style={thick,black}]
        \def\r{3}
    	\draw[cube] (0,0,0) -- (0,\r,0) -- (\r,\r,0) -- (\r,0,0) -- cycle;
    	\draw[cube] (0,0,\r) -- (0,\r,\r) -- (\r,\r,\r) -- (\r,0,\r) -- cycle;
    
    	\draw[cube] (0,0,0) -- (0,0,\r);
    	\draw[cube] (0,\r,0) -- (0,\r,\r);
    	\draw[cube] (\r,0,0) -- (\r,0,\r);
    	\draw[cube] (\r,\r,0) -- (\r,\r,\r);
    
        \draw[thick, fill=red!45] (0,0,\r) circle (3pt) node[anchor=south east] {$000$};
        \draw[thick, fill=white] (0,\r,\r) circle (3pt) node[anchor=south east] {$001$};
        \draw[thick, fill=white] (\r,0,\r) circle (3pt) node[anchor=south east] {$100$};
        \draw[thick, fill=blue!45] (\r,\r,0) circle (3pt) node[anchor=south east] {$111$};
        \draw[thick, fill=white] (0,\r,0) circle (3pt) node[anchor=south east] {$011$};
        \draw[thick, fill=white] (\r,\r,\r) circle (3pt) node[anchor=south east] {$101$};
        \draw[thick, fill=white] (\r,0,0) circle (3pt) node[anchor=south east] {$110$};
        \draw[thick, fill=white] (0,0,0) circle (3pt) node[anchor=south east] {$010$};
    
        \draw[very thick,-latex, opacity=0.7,bend right, gray!70!orange] (0,0,\r) to (\r,0,\r);
        \draw[very thick,-latex, opacity=0.7,bend right, gray!70!orange] (\r,0,\r) to (\r,0,0);
        \draw[very thick,-latex, opacity=0.7,bend right, gray!70!orange] (\r,0,0) to (\r,\r,0);
    
        \draw[very thick,-latex, opacity=0.7,bend right, gray!70!purple] (0,0,\r) to (0,0,0);
        \draw[very thick,-latex, opacity=0.7,bend right, gray!70!purple] (0,0,0) to (0,\r,0);
        \draw[very thick,-latex, opacity=0.7,bend right, gray!70!purple] (0,\r,0) to (\r,\r,0);
    \end{tikzpicture}
    \caption{Two paths (purple and orange) from $x=000$ to $y=111$ on the $3$-dimensional hypercube.}
    \label{Draw a 3/4 bit hypercube, x=000, y=111, draw all the path with different colors}
\end{figure}
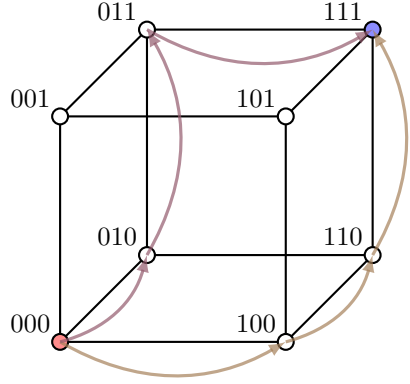
Lemma~\ref{Haar random 5 path} proved the satisfaction of the relaxed version of \textbf{Expansion}.
We now bound the probability that the path has all internal vertices good. $p= \Pr[c_l \leq z(x) \leq c_u] $, we have 
\begin{equation*}
    \Pr(\text{length 5 path internally all good}  )=p^4 =q
\end{equation*}
Among $N$ paths, we apply Corollary~\ref{tail bound Bernoulli}, we have
\begin{equation*}
    \Pr(\#\text{length 5 path internally all good} \leq \alpha N ) \leq 2^{-\frac{ 2(q-\alpha)^2N}{\ln 2}}
\end{equation*}
There are at most $2^{2n}$ pairs of $x,y$. Applying union-bound over all pairs yields the failure probability at most $2^{2n} 2^{-\frac{ 2(q-\alpha)^2N}{\ln 2}} $, which is exponentially small for $k>1$.

A stationary distribution $\pi(x)$ induced by Haar random states satisfies $k$-GLEP with high probability, when $k>1$. Hence, Haar-random states are covered by our fidelity estimation protocol.

\subsection{Approximate \texorpdfstring{$t$}--Designs}
\label{t-design proof}
Haar random states, while mathematically powerful, are highly entangled and physically difficult to prepare. States $t$-design, especially $\epsilon$-approximate state $t$-design, are more practically relevant. 
A unitary $t$-design is an ensemble $\mathcal{U}$ of unitaries such that randomly selecting a unitary $U \sim \mathcal{U}$, we have: for every $k_t \leq t$, any $\rho$,  $\mathbb{E}_{U \in \mathcal{U} }[U^{\otimes k_t} \rho (U^{\dagger})^{\otimes k_t}]=\mathbb{E}_{U \sim {\rm Haar} }[U^{\otimes k_t} \rho(U^{\dagger})^{\otimes k_t}  ] $. Consider state $t$-design $\vert \psi \rangle := U \vert 0 \rangle^{\otimes n}$, $\pi(x)  := \vert \langle x\vert \psi \rangle \vert^2$, for every $k_t \geq 0$, the average projector onto a random pure state~\cite{Haar_intro},
\begin{equation}
\label{moment k}
    \mathcal{M}_{k_t}:= \mathbb{E}_{\vert \psi \rangle \sim {\rm Haar} } [\vert \psi \rangle \langle \psi \vert^{\otimes k_t} ]=\frac{\prod^{(k_t)}_{\rm sym}  }{\binom{D+k_t-1}{k_t}}
\end{equation}
So, we have
\begin{equation}
\label{t-design 1}
    \mathbb{E}_{\mathcal{U} }[\pi(x)^{k_t}] = \langle x \vert^{\otimes k_t} \mathcal{M}_{k_t} \vert x \rangle^{\otimes k_t} = \frac{\langle x \vert^{\otimes k_t} \prod^{(k_t)}_{\rm sym} \vert x \rangle^{\otimes k_t}}{\binom{D+k_t-1}{k_t}} 
\end{equation}
where $\prod^{(k_t)}_{\rm sym}$ is the projector onto the symmetric subspace of the $k_t$-fold tensor product space (subspace of states that remain unchanged under any permutation of the $k_t$ subsystems). $\prod^{(k_t)}_{\rm sym}$ acts as the identity on the one-dimensional subspace spanned by $\vert x \rangle^{\otimes k_t}$, $\langle x \vert^{\otimes k_t} \prod^{(k_t)}_{\rm sym} \vert x \rangle^{\otimes k_t}=1$. $\binom{D+k_t-1}{k_t}=\frac{D(D+1) \cdots (D+k_t-1)}{k_t!}$. When $k_t \ll D$, $D(D+1) \cdots (D+k_t-1)=D^{k_t}(1+\mathcal{O}(k_{t}^2/D)) \approx D^{k_t} $.
Therefore, 
\begin{equation}
\label{t-design}
    \mathbb{E}_{\mathcal{U} }[\pi(x)^{k_t}] =\frac{k_t !}{D^{k_t}}
\end{equation}
\begin{theorem}[Markov’s Inequality]
    For any nonnegative random variable $X$ and any number $a$
    \begin{equation}
    \label{Markov ineq 1}
        \Pr [X \geq a ] \leq \frac{\mathbb{E}[X]}{a}
    \end{equation}
\end{theorem}

We first prove the Weight bound. Similar to Section~\ref{Haar random proof} we replace $\vert \supp{(\pi)} \vert$ in Definition~\ref{k-GLEP} and~\ref{Good Vertex} with $2^n$, we $\vert \supp{(\pi)} \vert= 2^n$ with high probability for the $t$-design we prove below. 

\begin{corollary}
\label{Markov’s Inequality k}
    For any nonnegative random variable $X$ and any number $a$
    \begin{equation}
    \label{Markov ineq k}
        \Pr [X \geq a ] \leq \frac{\mathbb{E}[X^{k_t}]}{a^{k_t}}
    \end{equation}
\end{corollary}

\begin{proof}
    Define the indicator of the event we care about $\mathbf{1}_{ \{X \geq a \}} =\begin{cases}
        1 & \text{if} \ \ X \geq a \\
        0 & \text{otherwise} 
    \end{cases}$. For every realisation of $X \geq 0$, if $X \geq a$, we have $\mathbf{1}_{ \{X \geq a \}}=1 $ and $X^{k_t} \geq a^{k_t}$, the ratio $\frac{X^{k_t}}{a^{k_t}}\geq 1$, so we have $\mathbf{1}_{ \{X \geq a \}} \leq \frac{X^{k_t}}{a^{k_t}} $. Take the expectations, we have $\Pr[X \geq a]=\mathbb{E} [\frac{X^{k_t}}{a^{k_t}}]=\frac{\mathbb{E}[X^k] }{a^k} $. 

    Another way to prove $Y:= X^{k_t}$, substitute $X$ in Equation~\ref{Markov ineq 1} with $Y$, $a$ with $a^{k_t}$, we have Equation~\ref{Markov ineq k}.
\end{proof}
With Equation~\ref{Markov ineq k} and~\ref{t-design}, set $X= \pi(x)$, $a= c_u' 2^{-n} $, $k_t=t$, we obtain:
\begin{equation}
\label{Markov’s Inequality upper bound}
    \Pr[ \pi(x) \geq c_u' 2^{-n} ] \leq \frac{t!}{D^t (c_u' 2^{-n})^t } =\frac{t!}{(c_u')^t}
\end{equation}
We set the failure probability of the Weight bound as $\delta=\frac{t!}{(c_u')^t}$, so we have
\begin{equation}
\label{c_u' 1}
    c_u'= (\frac{1}{\delta})^{1/t} (t!)^{1/t}
\end{equation}
Stirling approximation gives: $t!=\sqrt{2\pi t} (\frac{t}{e})^t \Big (1+ \mathcal{O}(\frac{1}{t}) \Big ) $, $(\sqrt{2\pi t})^{1/t}=(2\pi t)^{1/2t }=e^{\frac{\ln(2 \pi t) }{2t}}=1 + \mathcal{O}(\frac{\ln t}{t}) $, take the $1/t$-th power of $t!$ yields 
\begin{equation}
\label{t}
    (t!)^{1/t}=\Big (\frac{t}{e } \Big ) \Big (2 \pi t \Big)^{1/(2t)} \Big (1 + \mathcal{O}(\frac{1}{t}) \Big) =\frac{t}{e} \Big( 1 + \mathcal{O }(\frac{\ln t}{t}) \Big )
\end{equation}
We need an exponentially small failure probability, that is, $\delta=2^{-c'n}$, $c'>0$ is a constant. Take the $1/t$-th power of $\delta^{-1}$ yields
\begin{equation}
\label{delta}
    \delta^{-1/t}=2^{c'\frac{n}{t} }
\end{equation}
Substitute the corresponding terms in Equation~\ref{c_u' 1} with equations~\ref{t} and~\ref{delta}, we have 
\begin{equation}
\label{c_u' 2}
    c_u'= 2^{c'\frac{n}{t} } \Big (\frac{t}{e} \Big) \Big( 1 + \mathcal{O }(\frac{\ln t}{t}) \Big )
\end{equation}
The \textbf{Weight bound} condition is meant to exclude exponentially large weight, so, $ 2^{c'\frac{n}{t} }$ term in Equation~\ref{c_u' 2} is undesirable. To avoid exponentially large weight, $t$ must scale linearly with $n$, namely $t=\Theta(n)$, then $ 2^{c'\frac{n}{t} }$ becomes a constant $C_t$. Therefore, we have $c_u'= \Theta(t)$, $t=\Theta(n)$. 

The proof of Smoothness condition using Corollary~\ref{tail bound Bernoulli} in Section~\ref{Haar random proof} also applies to state $t$-design: when $m>1$, the failure
probability of \textbf{Smoothness} is exponentially small. The Smoothness condition requires that $ c_l 2^{-n} \leq \pi(x) \leq  c_u 2^{-n}$ satisfies with reasonably high probability: the failure probability needs to be polynomially small, i.e. $\mathcal{O}(n^{-\beta})$, $\beta > 1$, so that over the $\mathcal{O}(n^{k})$ neighbours, there are at least $\Omega(1)$ fraction of good neighbours for every vertex. We derive the requirements for $c_u, c_u$ to achieve this polynomially small failure probability.
For the upper bound, we apply Markov’s Inequality similarly to Equation~\ref{Markov’s Inequality upper bound}, then use Stirling approximation, we have
\begin{equation}
\label{Markov’s Inequality smooth upper bound}
    \Pr[ \pi(x) \geq c_u 2^{-n} ] \leq \frac{t!}{(c_u)^t} \approx \Big (\frac{t}{e } \Big )^t \Big (\frac{1}{c_u} \Big )^t \Big (1 + o(t) \Big ) 
\end{equation} 
Solving for, $\Big (\frac{t}{e } \Big )^t \Big (\frac{1}{c_u} \Big )^t \Big (1 + o(t) \Big ) = \mathcal{O}(n^{-\beta}) $: take the $1/t$-th power of both sides, we have $c_u \approx \frac{t}{e}n^{\beta/t} $. 
$n^{\beta/t}=e^{\frac{\beta \ln n}{t}} \approx 1+ \frac{\beta \ln n}{t} $, so $c_u=\Big (\frac{t}{e } \Big )\Big (1+ \frac{\beta \ln n}{t} \Big ) $, namely $c_u = \Theta( \ln n)$. (For the choice of minimal $c_u$, we do not consider $t=\Theta(n)$, but consider the dependence on $n$ because this condition itself does not require $t=\Theta(n)$. However, the $k$-GLEP actually require $t=\Theta(n)$, then $c_u = \Theta(t)$.)

To lower-bound $\pi(x)$, we apply Paley–Zygmund inequality.
\begin{theorem}[Paley–Zygmund inequality]
    For any nonnegative random variable $X$ and any $0 < \theta <1$, 
    \begin{equation*}
        \Pr[ X> \theta \mathbb{E}[X] ] \geq (1-\theta)^2 \frac{(\mathbb{E}[X])^2}{\mathbb{E}[X^2]}
    \end{equation*}
\end{theorem}
According to Equation~\ref{t-design}, $ \mathbb{E}[\pi(x)]=2^{-n} $, $\mathbb{E}[\pi(x)^2]=\frac{2!}{2^{2n} }$. Set $\theta=c_l$, the failure probability is 
\begin{equation*}
    \Pr[\pi(x) \leq c_l 2^{-n} ] < 1-(1-c_l)^2 = 2c_l-c_l^2
\end{equation*}
$c_l =\mathcal{O}(n^{-\beta'}) $ will yield the failure probability in $\mathcal{O}(n^{-\beta}) $. 

Therefore, the probability $p= \Pr[c_l 2^{-n} \leq \pi(x) \leq c_u 2^{-n}] $ is reasonably high, when $c_u = \Theta(\ln n)$ and $c_l = \mathcal{O}(n^{-\beta'})$. Then applying Corollary~\ref{tail bound Bernoulli} as in Section~\ref{Haar random proof} proves the \textbf{Smoothness} condition holds with failure probability of $2^{-\frac{ 2(p-\alpha)^2N}{\ln 2}}$ which is exponentially small when $k > 1$.

The \textbf{Expansion} condition follows as in Section~\ref{Haar random proof}. For $t = \Theta(n)$, the graph $G$ induced by state $t$-design has full support over the hypercube with high probability.

State $t$-design satisfies $k$-GLEP with high probability, when $t=\Theta(n)$, $m>1$.

The extension to $\epsilon$-approximate $t$-designs is straightforward; we omitted it here for conciseness and leave it to the dedicated readers. See \cite[Lemma 9]{Brand_o_2016} and \cite[Lemma 13]{Mann2017OnTC} for relevant bounds on probability weights.

\subsection{States Prepared by Random Low-Depth Circuits}
\label{random low-depth circuits}
For current quantum devices, random quantum circuits are among the most experimentally accessible applications, and they have been used in benchmarking or demonstrating quantum supremacy. A natural question is: do states prepared with such circuits of arbitrary depth satisfy $k$-GLEP, and thus fall within the scope of our protocol?

Recent progress shows that $\epsilon$-approximate unitary $t$-designs can be generated in depth $\mathcal{O}(\log (n)$ by gluing local designs prepared on polylogarithmic-size patches, with the dependences on $t$ and $\epsilon$ inherited (polylogarithmically) from the local design primitives used on the patches~\cite{schuster2025randomunitariesextremelylow}. With long-range two-qubit gates and $\mathcal{O}(nt)$ ancilla qubits, one can achieve $\epsilon$-approximate $t$-designs in depth $\mathcal{O}(\log t \cdot \log \log (nt/ \epsilon)) $~\cite{cui2025unitarydesignsnearlyoptimal}. Since we have proved that $\epsilon$-approximate state $t$-designs satisfy $k$-GLEP, random-circuit depth sufficient to realize such designs (e.g., for $t= \Theta(n)$) places the resulting states within the scope of our protocol.

However, if the random circuits only exhibit anticoncentration, that is, their output distribution is spread relatively evenly across outcomes without being overly concentrated on a small subset, do they fall into the scope of the protocol? Anticoncentration is weaker than $t$-design and requires lower depths: for generic random two-local circuits, anticoncentration requires $\Theta(\log n) $ depth~\cite{PRXQuantum.3.010333}. We believe that states prepared by random circuits exhibiting anti-concentration also lie within the scope of our protocol, but we leave a rigorous proof of this as an open question.

\subsection{Ground State of Gapped Local Hamiltonian}
\label{Hamiltonian}
A Hamiltonian is referred to as \emph{$k$-local} if it can be expressed as a sum of terms, each acting on a bounded (constant)  number $k$ of qubits.
When the off-diagonal components of a Hamiltonian, in the computational basis, are real-valued and non-positive, the Hamiltonian is said to be \emph{stoquastic}.
A Hamiltonian is called \emph{gapped}, if the spectral gap $\gamma$, the difference between the ground state energy and the first excited state energy, is bounded below, away from zero. For a $k$-local Hamiltonian $H$, we denote the \emph{ground state} and \emph{ground state energy} as $\ket{\phi_0}$ and $\lambda_0$ respectively.
We will additionally assume that $\norm{H} \leq {\rm poly}(n)$.

It is well-known that estimating the ground state energy of a general local Hamiltonian is QMA-complete~\cite{KSV02}, even extending to physically relevant systems~\cite{PM17}.
Though it has been demonstrated that when provided with additional input, a state having good overlap with the ground state, the ground state energy estimate problem becomes tractable for quantum computers~\cite{waite2025physically, waite2025guided, CKFH+23, GLG22}.
While it remains an open problem to develop constructive algorithms that produce such states, authenticating their form is equally important.
In this section, we demonstrate how Protocol~\ref{prot:fidelity} can be extended (beyond stoquastic Hamiltonians) for the fidelity certification between a produced lab state and the ground state of a general gapped local Hamiltonian.

Huang, Preskill, and Soleimanifar~\cite{10756060} applied the work of Bravyi, Gosset, and Liu~\cite{PhysRevLett.128.220503} that proved the mixing time of discrete-time Markov chains (DTMCs) constructed from the unique ground state of a local stoquastic Hamiltonian with gap $\gamma$ and sensitivity $s$ is bounded by 
\begin{equation*}
    \tau \leq \frac{2Ns}{\gamma}
\end{equation*}
where $s=\max_{x \neq y}\frac{\vert \langle y \vert H \vert x \rangle \langle x \vert \phi_0 \rangle \vert}{\langle y \vert \phi_0 \rangle \vert} $ is the sensitivity parameter for a given Hamiltonian $H$ with ground state $\vert \phi_0 \rangle$, $\gamma \geq \frac{1}{\textnormal{poly}}$ is the spectral gap, and $N= \sum_{i=1}^k \binom{n}{i} = \mathcal{O}(n^k)$.

Bravyi, Gosset, and Liu~\cite{PhysRevLett.128.220503} showed that, for stoquastic Hamiltonians, $s \leq \max_y \langle y \vert H \vert y \rangle-E_0 $, $E_0$ is the ground state energy of the system, so with $s$ and $\frac{1}{\gamma}$ scaling polynomially in $n$, the mixing time is $\mathcal{O}(\textnormal{poly}(n))$. We note that the upper bound on the sensitivity parameter depends both on the Hamiltonian norm and the ratio of the ground state's amplitudes; the result, therefore, does not hold for arbitrary ground states. Huang, Preskill, and Soleimanifar used this bound on mixing time $\tau$ to show that their protocol applies to the unique ground state of a local stoquastic Hamiltonian, as long as the gap $\frac{1}{\gamma}$ and sensitivity $s$ are polynomially bounded.

Bravyi \textit{et. al} used the fixed-node Hamiltonian construction~\cite{Haaf1995proof} to transform a general local Hamiltonian $H$ into a stoquastic Hamiltonian $F$, while preserving ground state properties~\cite{Bravyi_2023}. They proved the ground states of $F$ and $H$ have ground states related via the well-known isometry $i \equiv \big(\begin{smallmatrix}0 & -1\\1 & 0\end{smallmatrix}\big)$; an additional cost is an increase to the locality~\cite{waite2025succinct}.
Specifically, let $\ket{\phi_0} = \sum_x (a_x + i b_x)\ket{x}$ be the ground state of $H$, then the ground state of the corresponding stoquastic Hamiltonian $F$ is defined as $\ket{\psi_0} = \sum_x a_x\ket{x}\ket{0} + b_x\ket{x}\ket{1}$. It follows that the ground state energy of $F$ is $\lambda_0$ and the spectral gap is at least $\gamma$. We refer to Ref.~\cite{Bravyi_2023} for the detailed fixed-node Hamiltonian construction of $F$ and the proof. Bravyi~\textit{et. al} use a continuous-time Markov chain (CTMC) approach in conjunction with Gillespie’s algorithm to prove that the mixing time is proportional to $\frac{1}{\gamma}$. We refer to Theorem 1 and Theorem 2 in Ref.~\cite{Bravyi_2023} for the details and proofs. That is, with $\frac{1}{\gamma}$ and $\Vert H \Vert$ polynomially bounded, the mixing time of the Markov chain corresponds to the unique ground state $\vert \phi_0 \rangle$ is polynomially bounded and $\vert \phi_0 \rangle$ is the ground state of a general gapped $k$-local Hamiltonian $H$ that is not necessarily stoquastic. Protocol~\ref{prot:fidelity} does not require running the Markov chain: the efficiency and accuracy of Protocol~\ref{prot:fidelity} only require the mixing time to be polynomially bound. We directly apply the result of Bravyi \textit{et. al}~\cite{Bravyi_2023} to extend the scope of Protocol~\ref{prot:fidelity} to the ground states of more general gapped local Hamiltonians that are not necessarily stoquastic, given that $\frac{1}{\gamma}$ and $\Vert H \Vert$ are polynomially bounded.

We do not prove with $k$-GLEP for this example because there exists sufficient research \cite{Bravyi_2023, PhysRevLett.128.220503} to directly bound the mixing time, which is all we need to guarantee the performance of the fidelity estimation protocol.

\subsection{W States \& Dicke states}
\label{W states}
W states are entangled quantum states characterized by the equal superposition of all basis states with a single excitation, 
\begin{equation*}
    \vert W_n \rangle = \frac{1}{\sqrt{n}} \sum_{\substack{x \in \{0,1\}^n \\ |x|_1 = 1}} \vert x \rangle,
\end{equation*}
where $|x|_1$ refers to the $L^1$-norm (Hamming weight). These states have physical relevance~\cite{weng2025high} and are robust against particle loss; bipartite entanglement is retained even when one qubit is traced out. W states demonstrate that $k$ must be larger than 1 for our protocol to be applicable because the neighboring vertices $x, y \in \supp{(\pi)}$ have Hamming distance $d(x, y)=2$.

The weight bound and smoothness conditions are straightforward once we write out $\pi(x)  := \vert \langle x\vert W_n \rangle \vert^2$: 
\begin{equation*}
    \pi (x) =\begin{cases}
        \frac{1}{n}, &\text{if }  \vert x \vert_1=1, \\
        0, &\text{otherwise }
    \end{cases}
\end{equation*}
For W states, $\vert \supp{(\pi)} \vert= n$. For every vertex $x$ of the graph induced by W state, $\pi(x) = \frac{1}{n}=\frac{1}{\vert \supp{(\pi)} \vert} $, so for $c'_u \geq 1$, the \textbf{Weight bound} holds.

For any $0 < c_l \leq 1$ and $c_u \geq 1$, $\frac{c_l}{\vert \supp{(\pi)} \vert} \leq \pi(x) \leq \frac{c_u}{\vert \supp{(\pi)} \vert} $, that is, all the vertices in the induced graph are good, so \textbf{Smoothness} condition holds. We note that for such a graph, any two vertices (all the neighbors) are reachable with 2 bit-flip, so $k=2$. W states are typical examples that $k=1$ cannot suffice. 

The path construction is similar to that in the proof of Lemma~\ref{Haar random 5 path} except that the support is only on the $n$ of the vertices on the hypercube and $k=2$, that is, we flip two bits at each step. And the $n$ vertices are all good vertices, so the \textbf{Expansion} condition holds. 

We can further generalize the proof to $n$-qubit Dicke state with $k$ excitation
\begin{equation*}
    \vert D_k^n \rangle = \frac{1}{\sqrt{\binom{n}{k}}} \sum_{\substack{x \in \{0,1\}^n \\ |x|_1 = k}} \vert x \rangle,
\end{equation*}
The \textbf{Weight bound} and \textbf{Smoothness} conditions are straightforward, similar to the W states we discussed above. And the Expansion condition can also be proved in the same way. The same arguments apply when the state is
\begin{equation*}
    \vert \psi_k^n \rangle =  \sum_{\substack{x \in \{0,1\}^n \\ |x|_1 = k}} \alpha_x \vert x \rangle,
\end{equation*}  
as long as the weights $\alpha_x$ (or $\pi(x):= \vert \alpha_x \vert^2$) are approximately uniformly distributed.

We note that the graph structure for these states is no longer the full hypercube, like Haar or $t$-designs, but rather the Johnson graph $J(n, k)$, whose vertices are the $n$-bit string of Hamming weight $k$, and two vertices are connected if and only if they differ by exchanging a single 1 and a single 0 (i.e. Hamming distance 2) Such graph structure is known to have expansion properties.

\section{Empirical Check}
\label{Empirical check}
This section analyzes the sample complexity of Protocol~\ref{Prot: Empirical Check} and clarifies the procedure. We first bound the sample size parameters $S$, $M$, and $R$ via the Hoeffding inequality and a union bound over the three condition testings. We then clarify the ``Make samples'' step, and discuss the parameter $C$. $S$, $M$, $R$ and $C$ together give the sample complexity of Protocol~\ref{Prot: Empirical Check}. We also explain the ``Estimate $\vert \supp{\pi} \vert$'' step in detail. Throughout, $\epsilon, \delta \in (0, 1)$ denote accuracy and confidence parameters.

Let $p_{\rm weight} $, $p_{\rm smooth} $, $p_{\rm expand} $ be the true acceptance probabilities for the weight bound, smoothness, and expansion tests,  $\hat{p}_{\rm weight} $, $\hat{p}_{\rm smooth} $, $\hat{p}_{\rm expand} $ be the empirical acceptance fractions over the $S$ seed samples.
We use $\hat{p}_{\rm weight} $ as an example for the complexity analysis of $S$.
We determine that the weight bound is satisfied if $p_{\rm weight} \geq 1-\epsilon$, but we can only observe $\hat{p}_{\rm weight}$, if the sampling error $\vert \hat{p}_{\rm weight} - p_{\rm weight} \vert \leq \frac{\epsilon}{2} $, then $\hat{p}_{\rm weight} \geq 1-\frac{\epsilon}{2}$ ensures $p_{\rm weight} \geq 1-\epsilon$. For each condition testing, if sample $x_i$ is accepted is Bernoulli, by Hoeffding inequality in Theorem~\ref{Hoeffding’s inequality}, 
\begin{equation*}
    \Pr [\vert \hat{p}_{\rm weight}-p_{\rm weight} \vert > \frac{\epsilon}{2} ] \leq 2 \exp \big (-2S \cdot(\frac{\epsilon}{2})^2 \big )=2 \exp \Big (-\frac{S \epsilon^2}{2} \Big )
\end{equation*}
By a union bound over the three conditions, we require $2 \exp \Big (-\frac{S \epsilon^2}{2} \Big ) \leq \frac{\delta}{3}$, thus
\begin{equation*}
    S \geq \frac{2}{\epsilon^2} \ln \Big (\frac{6}{\delta} \Big)
\end{equation*}
$M$ and $R$ can be bounded with the same logic, but note that a union applies over the $S$ sample, so they are bounded as 
\begin{equation*}
    M \geq \frac{2}{\epsilon^2} \ln \Big (\frac{6S}{\delta} \Big), \ \ R \geq \frac{2}{\epsilon^2} \ln \Big (\frac{6S}{\delta} \Big)
\end{equation*}
For the smoothness condition, we can visit all neighbors; in that case, $M=N= \mathcal{O}(n^k)$. For the expansion condition, when we check/reconstruct the short path, we need to visit $\mathcal{O}(n^k)$ vertices. Checking if there are ``$\alpha N$ pairwise internally-disjoint paths of length at most 5 with all good vertices'' is not always efficient, so we also propose a relaxed but sufficient check: The Expansion condition can be relaxed, as we noted in Section~\ref{Conditions for Fast Mixing}, we can thus search to find one short good-only path, which can be done using randomized breadth-first searches (BFS) restricted to a local $k$ neighborhood and bounded depth $d=5$. Each BFS trial require $\textnormal{poly}(N)$ queries, $N \approx n^k$.

We now clarify the ``Make samples'' step in Protocol~\ref{Prot: Empirical Check}. If we have physical (quantum) copies of $\vert \psi \rangle$, each measurement on the computational basis directly produces an exact sample from $\pi$, we do not need the ``Make samples'' step. Nevertheless, assuming having physical copies is not practical, we only assume query access to $\pi$, which is directly available via the query access to the amplitude $\Psi(x) = \langle x | \psi \rangle$ that is assumed by Protocol~\ref{prot:fidelity}. 

\begin{prop}[Rejection sampling]
\label{Rejection sampling}
Let $q(x)=2^{-n} $ on $\{0, 1 \}^n$ and let $2^n \geq C \geq \sup_x \frac{\pi(x)}{q(x)}$. Draw $x_U \sim q$ and accept it with probability $\frac{\pi(x_U)}{C q(x_U)}$. Then the distribution of accepted samples is exactly $\pi$.
\end{prop}
\begin{proof}
    For any event $E \subseteq \{0, 1 \}^n$
\begin{equation*}
    \Pr[x_U \in E \text{ and accepted} ]=\sum_{x \in E} q(x) \cdot \frac{\pi(x)}{C q(x) } = \frac{1}{C} \sum_{x \in E} \pi(x)
\end{equation*}
The unconditional acceptance probability per $x_U$ is $p:= \frac{1}{C}$, so
\begin{equation*}
    \Pr[ \text{accepted } x_U \in E \vert \text{ accepted}] =\sum_{x \in E} \pi(x)
\end{equation*}
i.e., the accepted sample is exactly $\pi$-distributed.
\end{proof}
Apply Proposition~\ref{Rejection sampling}, with only query access to $\pi$, we can make samples that represent $\pi$ via rejection sampling \cite{devroye2006nonuniform, von195various}. As the acceptance rate is $\frac{1}{C}$, this make samples step requires $\mathcal{O}(CS)$ samples (queries). 

In protocol~\ref{Prot: Empirical Check}, we assume $C$ is given, similar to the rejection sampling in~\cite{ von195various}. $C$ can be exponential in $n$, e.g., $C=2^n$ for a basis state; $2^{n-1} $ for GHZ; $\frac{2^n}{n}$ for W states.
For more uniform-like distribution, $C= \textnormal{poly}(n)$, e.g., for Haar random states, $C=\Theta(n)$ with high probability; for state $t$-design, $C= \mathcal{O}(n)$ with high probability. We note that the basis state and GHZ state are hard examples that are beyond the scope of protocol~\ref{prot:fidelity}. For W states, Dicke states or Hamming weight states, though $C$ is exponential in $n$, thus Protocol~\ref{Prot: Empirical Check} is not efficient, directly verifying if a target state belongs to such ensembles is feasible. In fact, for states with such small support ($\textnormal{poly}(n)$), they only fall into the scope of Protocol~\ref{prot:fidelity} if the weights are (almost) evenly distributed. In short, if a given $C$ is exponentially big, we can directly declare \textsc{Fail} for the ``Make samples'' step and verify if the target state has weights that are (almost) evenly distributed on a polynomial support (i.e., belong to the ensembles of $W$ states, Dicke states or Hamming weight states) and decide if Protocol~\ref{prot:fidelity} is applicable.

If $C$ is not given, we can find a $C$ adaptively with the following procedure, similar to standard Monte Carlo:
\begin{itemize}
    \item Set an initial $C_0$ (e.g., $C_0=c_u'$).
    \item For $x_U \sim \textnormal{Unif} ( \{0, 1 \}^n) $, if $ \frac{\pi(x_U)}{C 2^{-n}} \leq 1 $, accept $x_U$ with probability $\frac{\pi(x_U)}{C 2^{-n}}$; if $ \frac{\pi(x_U)}{C 2^{-n}} \geq 1 $, set $C= \max \{ 2C, \frac{\pi(x_U)}{ 2^{-n}}\} $, discard all previous accepted seed samples, and restart the sampler. After at most $\ceil{\log_2 (\frac{C_{\max} }{C_0}) }$ increases, $C$ exceeds $C_{\max}:= \max_{x_U} \frac{\pi(x_U)}{ 2^{-n}}$, the sampler runs with a valid $C$; all seed samples collected are i.i.d. $\sim \pi$
\end{itemize}

Another option to make seed samples is simulating the Markov chain induced by $\pi$, but we prefer to avoid this because Protocol~\ref{prot:fidelity} essentially does not require running the Markov chain, and in some sense Protocol~\ref{Prot: Empirical Check} tries to check if the Markov chain is guaranteed to be fast mixing.

We now clarify the ``Estimate  $\vert \supp{\pi} \vert$'' step in Protocol~\ref{Prot: Empirical Check}: if $\vert \supp{\pi} \vert$ is not known a priori, we estimate it with the seed samples we make in the first step. Estimating the support size is normally equivalent to a population estimation problem, which can be approached via the reverse birthday paradox~\cite{4557157}: sample until $r$ repeated outcomes are seen, and then estimate the support size as roughly $\frac{S_p^2}{2r}$, where $S_p$ is the number of samples needed to observe $r$ repeats. The number of samples required for an accurate estimation is the square root of the population size (i.e., support size). Detailed proof and performance guarantee can be found in Ref.~\cite{4557157}. 

With $S$ samples, we can estimate an ``effective support'' using the collision estimator $\hat{C}_2$ with the same gist of reverse birthday paradox~\cite{4557157}. The probability that two independent points $x_i, x_j$ are equal is $C_2= \Pr[x_i= x_j]=\sum_x \pi(x)^2 $. If $\pi$ is (almost) uniform on its support, each $\pi(x) \approx \frac{1}{\vert \supp{\pi} \vert}$, it follows that $\sum_x \pi(x)^2 \approx \vert \supp{\pi} \vert \cdot \frac{1}{\vert \supp{\pi} \vert^2} = \frac{1}{\vert \supp{\pi} \vert}$. If $\pi$ is non-uniform, this does not exactly give the support size, but still reflects an ``effective support'', which is the part that carries most of the probability mass. 

If the support size is exponential in $n$, a polynomial number of samples cannot yield an accurate estimate for $\vert \supp{\pi} \vert$. In such cases, the collision estimator will typically output a large estimate for $\vert \supp{\pi} \vert$, but not necessarily an accurate one. If we suspect the support size is large, and the collision estimator may not yield a reliable estimation, one option is to compute the lower confidence bound for the collision estimator, thus yielding an upper bound for $\vert \supp{\pi} \vert$. When checking the $k$-GLEP conditions, an overestimate of $\vert \supp{\pi} \vert$ makes the weight bound and smoothness condition stricter, preventing a possible ``spike'' that impedes fast mixing, therefore, an upper bound of $\vert \supp{\pi} \vert$ provides a more conservative check and avoid false \textsc{Pass}. The collision estimator $\hat{C}_2=\frac{1}{S(S-1)  } \sum_{i \neq j} \mathbf{1} [ x_i=x_j ] $ can be analyzed with U-statistic Hoeffding’s inequality~\cite{serfling1980approximation, Ai2022HoeffdingSerflingIF}, 
\begin{equation*}
    \Pr( \vert \hat{C}_2 -C_2 \vert \geq t ) \leq 2 \exp ( -c_{uh} S t^2 )
\end{equation*}
$c_{uh}$ is an coefficient, we can take $c_{uh}=\frac{1}{2}$ or $\frac{1}{4}$ for example, fix $\delta$, we have $t=\sqrt{ \frac{1}{c_{uh}} \cdot\frac{ \ln (2/\delta)}{ S } } $, then with probability at least $1-\delta$, $C_2 \in [\hat{C}_2- t, \hat{C}_2+ t] $. Define the conservative estimate $\hat{C}_2^L = \max \{ \hat{C}_2-t, 0 \} $, $\frac{1}{\hat{C}_2^L} $ provides a conservative upper bound (overestimate) for $\vert \supp{\pi} \vert$.

If no collision was observed among all $S$ samples $x_1, \ldots, x_S  \sim \pi$, the support is likely to be large, possibly exponential in $n$. As $C_2= \sum_x \pi(x)^2$, among $S$ samples, there are $\binom{S}{2}$ pairs, the probability of no collision is upper bounded by \footnote{Although the $\binom{S}{2}$ pair are not independent, this still provides a valid upper bound.}
\begin{equation*}
    \Pr[\text{no collision}] \leq (1-C_2)^{\binom{S}{2}} \leq e ^{-C_2 \binom{S}{2} }
\end{equation*}
Set $e ^{C_2 \binom{S}{2}} \geq \delta$, then with confidence at least $1-\delta$, $C_2 \leq \frac{\ln (1/\delta)}{\binom{S}{2} }$, that is $\vert \supp{\pi} \vert \geq \frac{\binom{S}{2}}{\ln (1/ \delta) }$.

Alternatively, when no collisions are observed, the support size can also be inferred from the rejection sampling parameter $C$ in the ``Make samples'' step. To see this: Suppose all nonzero $\pi(x)$ are equal, i.e., $\pi$ is uniformly distributed on its support. Then $\frac{\pi(x)}{2^{-n}}= \frac{1/\vert \supp{\pi} \vert}{2^{-n}}=\frac{2^n}{\vert \supp{\pi} \vert} $, $C=\sup_x \frac{\pi(x)}{2^{-n}}=\frac{2^n}{\vert \supp{\pi} \vert}$, therefore $\vert \supp{\pi} \vert =\frac{2^n}{C}$. For non-uniformly distributed $\pi$, $C$ gives an upper bound in concentration; the actual support could be larger, but $\frac{2^n}{C}$ still serves as a lower bound.

The total sample complexity for Protocol~\ref{Prot: Empirical Check} is $CS  + S M + S \mathcal{O}(R)$, that is  $\mathcal{O} \Big (C \frac{1}{\epsilon^2} \ln \big (\frac{1}{\delta} \big) + \frac{1}{\epsilon^2} \ln \big (\frac{1}{\delta} \big)+ \frac{1}{\epsilon^4} \ln^2 \big (\frac{1}{\delta} \big ) \cdot \textnormal{poly} (n^k) \Big) $, $k$ is a small number, as long as $C$ is polynomially bounded, the query complexity does not grow exponentially in $n$. For the cases when $C$ is exponential, directly verifying certain state ensembles or replacing the rejection sampling with simulating the Markov chain (or given access to $S$ samples representing $\pi$) restores a polynomial complexity of the empirical check. This concludes the proof of Result~\ref{result:empirical-check}. We note that this empirical check certifies that the $k$-GLEP conditions hold on most of the probability mass; it does not guarantee they hold pointwise for all basis states. The empirical check is more of a diagnostic tool.

\section{Certifying Mixed States}
\label{Mixed states}
We consider a mixed target state
\begin{equation}
\label{fidelity  simple}
    \sigma = \sum_{i=1}^M p_i \vert \psi_i \rangle \langle \psi_i \vert, \ \ p_i \geq 0, \sum_{i=1}^M p_i=1, 1 \leq i \leq M
\end{equation}
where each $\{ \vert \psi_i \rangle \}$ is a pure state satisfying $k$-GLEP, thus lies within our protocol's scope. We define $f_i= \langle \psi_i \vert \rho \vert \psi_i \rangle $, which is the fidelity of $\rho$ against the pure state $\vert \psi_i \rangle$. And recall Uhlmann fidelity $F(\rho, \sigma)= \Vert \sqrt{\rho} \sqrt{\sigma} \Vert_1^2= [\Tr\sqrt{\sqrt{\rho} \sigma \sqrt{\rho} } ]^2 $. This section relates $F(\rho, \sigma)  $ to $f_i$ and proves Theorem~\ref{Mixed States theorem}.

$\Tr\sqrt{\sqrt{\rho} \sigma \sqrt{\rho} }$ is jointly concave (see~\cite[Theorem 9.7]{Nielsen_Chuang_2010} for a proof), so 
\begin{equation*}
    \Tr\sqrt{\sqrt{\rho} \sigma \sqrt{\rho} }=\sqrt{F(\rho, \sigma)}=\sqrt{F(\rho, \sum_i p_i \vert \psi \rangle \langle \psi_i \vert ) } \geq \sum_i^M p_i \sqrt{F(\rho, \vert \psi_i \rangle \langle \psi_i \vert) }= \sum_i^M p_i \sqrt{f_i} 
\end{equation*}
That is,
\begin{equation}
\label{mixed lower bound}
    F(\rho, \sigma)  \geq \Big ( \sum_{i=1}^M p_i \sqrt{f_i}  \Big )^2
\end{equation}
which provides the lower bound.

To obtain an upper bound, we use trace subadditivity for concave functions, that is, Rotfel’d trace inequality.
\begin{theorem}[Rotfel’d trace inequality]
\label{Rotfel’d trace inequality}
    If $A, B  \succeq 0$ and $f: [0, \infty) \rightarrow [0, \infty)$ is concave with $f(0) \geq 0$, then 
    \begin{equation*}
        \Tr f(A+B) \leq \Tr f(A) + \Tr f(B)
    \end{equation*}
\end{theorem}
Theorem~\ref{Rotfel’d trace inequality} is proved in Refs.~\cite{BOURIN2007512, article_Lee}. Let $f(t)=\sqrt{t}$, we have $\Tr \sqrt{A+B} \leq \Tr \sqrt{A}+ \Tr \sqrt{B} $. Set $X_i := p_i \sqrt{\rho} \vert \psi_i \rangle \langle \psi_i \vert \sqrt{\rho}$, then we find $\sqrt{\rho} \sigma \sqrt{\rho} =\sum_{i=1}^M X_i $. 
Each $X_i$ is a rank-1 positive semidefinite \footnote{Define a vector $\vert w_i \rangle:= \sqrt{\rho} \vert \psi_i \rangle$, then $X_i:= p_i \vert w_i \rangle \langle w_i \vert $, $X_i$ is rank-1. For any $\vert v \rangle$, $\langle v \vert X_i \vert v \rangle = p_i \vert \langle w_i \vert v \rangle \vert^2 \geq 0  $, $X_i$ is psd.}, so $\Tr(X_i)= p_i  \langle \psi_i \vert \sqrt{\rho} \sqrt{\rho} \vert \psi_i \rangle= p_i  \langle \psi_i \vert \rho \vert \psi_i \rangle=p_i f_i $, we have $\Tr \sqrt{X_i}=\sqrt{p_i} \Vert \sqrt{\rho} \vert \psi_i \rangle \Vert = \sqrt{p_i} \sqrt{\langle \psi_i \vert \rho \vert \psi_i \rangle}=\sqrt{p_i f_i} $. Applying Theorem~\ref{Rotfel’d trace inequality} yields 
\begin{equation*}
    \Tr\sqrt{\sqrt{\rho} \sigma \sqrt{\rho} } = \Tr \sqrt{ \sum_{i=1}^M X_i} \leq \sum_{i=1}^M \Tr \sqrt{X_i}= \sum_{i=1}^M \sqrt{p_i f_i},
\end{equation*}
thus we have
\begin{equation}
\label{mixed upper bound}
    F(\rho, \sigma) \leq \Big (\sum_{i=1}^M \sqrt{p_i f_i} \Big )^2
\end{equation}
Combining Equation~\ref{mixed lower bound} and~\ref{mixed upper bound} gives Theorem~\ref{Mixed States theorem}. In the special case where $[\rho, \sigma]=0$, namely $\rho$ and $\sigma$ commute, and $\sigma=\sum_{i=1}^M p_i \vert \psi_i \rangle \langle \psi_i \vert$ in a common eigenbasis with $\rho=\sum_{i=1}^M r_i \vert \psi_i \rangle \langle \psi_i \vert$, then $F(\rho, \sigma)= \Big (\sum_{i=1}^M \sqrt{r_i p_i} \Big)^2 $, which is the standard reduction of Uhlmann fidelity to the classical Bhattacharyya coefficient.
\end{appendices}
\printbibliography
\end{document}